\documentclass[11pt]{article}%
\usepackage{algorithm}
\usepackage{algpseudocode}
\usepackage{amssymb}
\usepackage{amsfonts}
\usepackage{amsmath}
\usepackage{graphicx}%
\providecommand{\U}[1]{\protect\rule{.1in}{.1in}}
\usepackage{tikz}
\usetikzlibrary{shapes}
\usepackage{pgfplots}
\usepackage{verbatim}
\usepgfplotslibrary{statistics}
\pgfplotsset{compat=1.8}
\usepackage{amsthm}
\renewcommand\and{\end{tabular}\kern-\tabcolsep\ and\ \kern-\tabcolsep\begin{tabular}[t]{c}}
\let\origthanks\thanks
\renewcommand\thanks[1]{\begingroup\let\rlap\relax\origthanks{#1}\endgroup}

\newtheorem{theorem}{Theorem}

\newtheorem{corollary}{Corollary}

\newtheorem{definition}[theorem]{Definition}

\newtheorem{lemma}[theorem]{Lemma}

\newtheorem{proposition}[theorem]{Proposition}

\begin{document}

\title{Network Reconstruction and Prediction of Epidemic Outbreaks for NIMFA Processes}
\author{Bastian Prasse\thanks{ Faculty of Electrical Engineering, Mathematics and
Computer Science, P.O Box 5031, 2600 GA Delft, The Netherlands; \emph{email}:
b.prasse@tudelft.nl, p.f.a.vanmieghem@tudelft.nl} \and Piet Van Mieghem\footnotemark[1]}
\date{Delft University of Technology\\
November 16, 2018}
\maketitle
\begin{abstract}
Predicting the viral dynamics of an epidemic process requires the knowledge of the underlying contact network. However, the network is not known for most applications and has to be inferred from observing the viral state evolution instead. We propose a polynomial-time network reconstruction algorithm for the discrete-time NIMFA model based on a basis pursuit formulation. Given only few initial viral state observations, the network reconstruction method allows for an accurate prediction of the further viral state evolution of every node provided that the network is sufficiently sparse.
\end{abstract}

\section{Introduction} 
 The field of epidemics encompasses a plethora of phenomena and is rooted in the description of infectious diseases \cite{anderson1992infectious}, with seminal works by Bernoulli \cite{bernoulli1760essai} and Snow \cite{snow1855mode}. Beyond infectious diseases, the spread of opinions, trends and fake news on online social networks can be described as the epidemic of a viral infection, whereby each individual is either infected (with the opinion, trend, etc.) or healthy. Epidemic processes over networks assume that the spreading may occur from one to another individual only if the two individuals have contact \cite{pastor2015epidemic}, for instance by a friendship relation. On a coarser level, one can describe the evolution of the virus between groups, or communities, of similar individuals, which is the focus of this work.
    
We consider the viral spread over a network with $N$ nodes, represented by an $N \times N$ adjacency matrix $A$, which specifies the existence of links between the nodes. The elements $a_{ij}$ of the adjacency matrix $A$ denote the presence and the absence of a link between node $i$ and $j$ by $a_{ij} = 1$ and $a_{ij} = 0$, respectively. The viral state of node $i$ at continuous time $t \ge 0$ is denoted by $v_i(t) \in [0, 1]$, which can be interpreted as the fraction of infected individuals of group $i$ at time $t$. The NIMFA model \cite{van2009virus, van2011n} describes the evolution of the viral state as
\begin{align}
	\frac{d v_i (t)}{d t } & = - \delta v_i(t) + \beta \sum^N_{j = 1} a_{i j} v_j(t) - \beta v_i(t) \sum^N_{j = 1} a_{i j}  v_j(t), \label{NIMFA_continuous}
\end{align}
for node $i = 1, ..., N$. The NIMFA model (\ref{NIMFA_continuous}) has two parameters: the \textit{curing rate} $\delta \ge 0$ and the \textit{infection rate} $\beta$. In this work, we confine ourselves to the \textit{discrete-time NIMFA model} \cite{pare2018analysis}
\begin{align}
v_i [k + 1] & = (1 - \delta_T) v_i[k] + \beta_T (1 - v_i[k])  \sum^N_{j=1} a_{i j} v_j[k], \quad i = 1, ..., N, \label{NIMFA_disc_}
\end{align}
where $k \in \mathbb{N}$ denotes the discrete time slot. The equations (\ref{NIMFA_disc_}) are obtained by applying Euler's method \cite{stoer2013introduction} to the continuous-time NIMFA (\ref{NIMFA_continuous}) with sampling time $T\ge 0$, and the spreading parameters of (\ref{NIMFA_disc_}) follow as $\delta_T = \delta T$ and $\beta_T = \beta T$. Equivalently to (\ref{NIMFA_disc_}), we can represent the discrete-time NIMFA model as vector equations
\begin{align}
v [k + 1] & = (1 - \delta_T) v[k] + \beta_T \textrm{diag}(u - v[k]) A v[k],\label{NIMFA_disc_stacked}
\end{align}
where the viral state vector equals $v[k] = (v_1 [k], ..., v_N [k])^T$ and $u$ is the all-one vector. Originally \cite{van2009virus}, the NIMFA model (\ref{NIMFA_continuous}) was proposed as an approximation of the susceptible-infected-susceptible (SIS) epidemic model \cite{pastor2015epidemic}. Here, we consider the NIMFA model (\ref{NIMFA_disc_}), and not the SIS model, as an exact description of the viral state evolution of a set of $N$ groups, which is justified by the validation of Par\'e \textit{et al.} \cite{pare2018analysis}. Our focus is the network reconstruction problem for NIMFA viral state observations: \textit{Given an observed viral state sequence $v[0], ..., v[n-1]$ of observation length $n$, can we estimate the adjacency matrix $A$ and the spreading parameters $\beta_T$, $\delta_T$? }

In Section \ref{sec:related_work}, we review related work. The nomenclature is introduced in Section \ref{sec:nomenclature}, and we briefly review sparse parameter estimation in Section \ref{sec:lasso}. We analyse the viral state dynamics for groups of individuals in Section \ref{sec:viral_dynamics}. The network reconstruction and spreading parameter estimation method is introduced in Section \ref{sec:network_reconstruction}. Numerical evaluations in Section \ref{sec:numerical_evaluation} demonstrate the performance of the network reconstruction method.

\section{Related Work}
\label{sec:related_work}
Par\'e \textit{et al.} \cite{pare2018analysis} analysed the equilibria of the discrete-time NIMFA model (\ref{NIMFA_disc_stacked}) and validated the dynamics of real-world epidemics, when the nodes of the network corresponds to groups of individuals, namely either households or counties.

Recently, estimation methods were proposed \cite{prasse2018exact, vajdi2017missing, shen2014reconstructing} to reconstruct the network from viral state observations of susceptible-infected-susceptible (SIS) epidemic models . The maximum-likelihood SIS network reconstruction problem is NP-hard \cite{prasse2018nphardness}, and the number of required viral state observations $n$ seems \cite{prasse2018exact} to grow (almost) exponentially with respect to the network size $N$. For the NIMFA model (\ref{NIMFA_disc_stacked}), Par\'e \textit{et al.} \cite{pare2018analysis} proposed a method to estimate the spreading parameters $\beta_T$ and $\delta_T$ under the assumption that the adjacency matrix $A$ is known exactly.

The network reconstruction method in this work is motivated by two factors. First, the tremendous number of required viral state observations and the NP-hardness seem to render the exact SIS network reconstruction hardly viable, and modelling the viral dynamics by the NIMFA equations (\ref{NIMFA_continuous}) may allow for a feasible network reconstruction problem. Second, we generalise the spreading parameter estimation method \cite{pare2018analysis} by also estimating the adjacency matrix $A$ of the underlying contact network.
 
\section{Nomenclature}
\label{sec:nomenclature} 
The $\ell_q$-``norm''\footnote{The notation $\lVert z \rVert_q$ is commonly used for any $q\ge 0$, but $\lVert z \rVert_q$ should not be mistaken for a norm if $0 \le q<1$ since $\lVert z \rVert_q$ does not satisfy the triangle inequality if $0 \le q<1$.} of an $N \times 1$ vector $z$ is given by
\begin{align*}
\left\lVert z \right\rVert_q = \left( \left| z_1 \right|^q + ... + \left| z_N \right|^q\right)^{1/q}
\end{align*}
for any $q>0$. For $q=0$, we denote the number of non-zero elements of an $N \times 1$ vector $z$ (the $\ell_0$-``norm'') by $\left\lVert z \right\rVert_0$, and we call a vector $z$ to be $s$-sparse if $z$ has at most $s$ non-zero components. The cardinality of a finite set $S$ is denoted by $|S|$. We denote the set of natural numbers which are smaller or equal to $i$ as $\mathbb{N}_i = \{1, 2, ..., i\}$. For an $m \times n$ matrix $M$, we denote the columns by $M_l \in \mathbb{R}^m$ for $l = 1, ..., n$, and the $mn \times 1$ vector $M_\textrm{vec} = (M^T_1, ..., M^T_n)^T$ is obtained by concatenating the columns of the matrix $M$ in one vector. The truncated singular value decomposition \cite{hansen1990truncated} of a real $m \times n$ matrix $M$ for some integer $r \le \textrm{min}\{m, n\}$ is given by $\tilde{M}_r = U_{r} S_{r} Q_{r}^T$ where $S_{r} = \textrm{diag}(\sigma_1, ..., \sigma_{r})\in \mathbb{R}^{r \times r}$ is the diagonal matrix with the $r$ greatest singular values of $M$, and the matrices $U_{r}\in \mathbb{R}^{m \times r}$ and $Q_{r}\in \mathbb{R}^{n \times r}$ are composed of the first $r$ left-singular and right-singular vectors of $M$, respectively.
 
A graph $G$ consists of $N$ nodes and $L$ links. The degree $d_i$ of a node $i$ equals $d_i = \sum^N_{j=1}a_{i j}$. We denote the largest eigenvalue of the adjacency matrix $A$ by $\lambda_1$. The effective infection rate is defined as $\tau = \frac{\beta}{\delta}$. The maximum number of links for an undirected graph with $N$ nodes and without self-loops equals $L_\textrm{max} = N(N-1)/2$. The $L_\textrm{max} \times 1$ link vector $x$ denotes the presence or absence of a link $l = 1, ..., L_\textrm{max}$ by $x_l = 1$ and $x_l = 0$, respectively, and the weighted link vector equals $w = \beta_T x$. Table \ref{table:nomenclature} summarises the nomenclature.

\begin{center}
  \begin{table}
  \centering
  \begin{tabular}{ | l | l | }
    \hline   
 $A \otimes B$ & Kronecker product of the matrices $A$, $B$ \\ \hline   
$\beta_T$ & Infection rate $\beta_T = T \beta$ \\ \hline    
$d_i$ & Degree of node $i$\\ \hline     
$\delta_T$ & Curing rate and sampled-time curing intensity $\delta_T = T \delta$ \\ \hline     
 $\textrm{diag}(x)$ & For a vector $x \in \mathbb{R}^N$, $\textrm{diag}(x)$ is the $N \times N$ diagonal matrix with $x$ on its diagonal \\ \hline   
$I_N$ & The $N \times N$ identity matrix \\ \hline       
    $L_\textup{\textrm{max}}$ & Maximum number of links for a network with $N$ nodes, $L_\textup{\textrm{max}} = N(N-1)/2$\\ \hline
    $\mathbb{N}_i$ & Set of natural numbers not greater than $i$, $\mathbb{N}_i = \{1, 2, ..., i\}$ \\ \hline
 $N$& Number of nodes\\ \hline
    $n$& Number of time instants which were observed, $n \in \mathbb{N}$\\ \hline
    $T$& Sampling time of the discrete-time NIMFA model\\ \hline
 $u$ & All-one vector $u = (1, ..., 1)^T \in \mathbb{R}^N$ \\ \hline    
$v[k]$ & Viral state at discrete time $k$, $v_i[k] \in [0, 1]$ \\ \hline
$V$ & Viral state matrix $V = (v[0], ..., v[n-1])^T \in \mathbb{R}^{n \times N}$ \\ \hline
$x$ & Link vector, $x_l = 1$ and $x_l = 0$ denote the presence or absence of link $l \in \mathbb{N}_{L_\textup{\textrm{max}}}$ \\ \hline
$\hat{w}(\epsilon)$ & Estimate of the link vector $x$ if the link-threshold equals $\epsilon$\\ \hline
$w$ & Weighted link vector $w = \beta_T x$ \\ \hline
    $\left\lVert z \right\rVert_q$ & $\ell_q$-``norm'' of the vector $z$ where $q \ge 0$\\ \hline
  \end{tabular}
  \caption{Nomenclature \label{table:nomenclature}} 
  \end{table}
\end{center}

\section{Sparse Parameter Estimation} 
 \label{sec:lasso}  
Consider the problem of estimating a parameter vector $\alpha = (\alpha_1, ..., \alpha_p)^T$ by a number $m$ of linear measurements or observations. The vector of measurements $y = (y_1, ..., y_m)^T$ follows from the linear system
\begin{align}\label{eq:lin_sys_cs}
y = X \alpha,
\end{align}
where $X$ is a given $m \times p$ measurement matrix with $p\ge m$. The rows of the measurement matrix $X$ are assumed to be linearly independent, so that $\textrm{rank}(X) = m$. If the number of measurements $m$ equals the number of elements $p$ of the parameter vector $\alpha$, then $\alpha$ can be obtained uniquely from the measurements $y$. 

In many applications, the following two facts hold regarding the linear system (\ref{eq:lin_sys_cs}). First, the number of linearly independent measurements $m$ is (possibly considerably) smaller than the number $p$ of parameters $\alpha_1, ..., \alpha_p$, and hence the linear system (\ref{eq:lin_sys_cs}) is underdetermined. Second, the parameter vector $\alpha$ often is sparse, i.e. the number of non-zero elements of $\alpha$ is significantly smaller than the total number $p$ of elements of $\alpha$. In the following, we introduce some selected results on sparse parameter estimation, which show that it is possible to estimate the parameter vector $\alpha$ accurately with a number $m<p$ of measurements by exploiting the sparsity of $\alpha$.

 If the parameter vector $\alpha$ is known to be sparse, then we can choose to estimate $\alpha$ as \textit{the sparsest} solution of the underdetermined linear system (\ref{eq:lin_sys_cs}) by solving 
\begin{align} \label{eq:best_subset_selection}
\hat{\alpha} = \text{arg }\underset{\alpha}{\text{min}}  \quad \left\lVert \alpha \right\rVert_0 \quad \text{subject to} \quad  y = X \alpha
\end{align}
 The estimation problem (\ref{eq:best_subset_selection}) suffers from the disadvantage that its objective function $\left\lVert \alpha \right\rVert_0$ is non-convex, and computing the solution to (\ref{eq:best_subset_selection}) in general is very difficult. A commonly employed approach is to replace the objective in (\ref{eq:best_subset_selection}) by the $\ell_1$-norm $\lVert \alpha \rVert_1$, which has two decisive advantages \cite{hastie2015statistical} over $\ell_q$-``norms'' with $q \neq 1$. First, in contrast to the non-convex $\ell_q$-``norms'' $\left\lVert \alpha \right\rVert_q$ with $0\le q<1$, the $\ell_1$-norm is convex, which enables solving the resulting optimisation problem efficiently \cite{boyd2004convex}. Second, among the convex $\ell_q$-norms $\left\lVert \alpha \right\rVert_q$ with $q\ge 1$, the $\ell_1$-norm is the best approximation of the $\ell_0$-``norm'' $\lVert \alpha \rVert_0$. 
 
Replacing the objective in (\ref{eq:best_subset_selection}) by the $\ell_1$-norm yields the optimisation problem
\begin{align} \label{eq:min_l1_norm_constrained_no_noise}
\hat{\alpha} = \text{arg }\underset{\alpha}{\text{min}}  \quad \left\lVert \alpha \right\rVert_1 \quad \text{subject to} \quad y  = X \alpha,
\end{align}
which is coined \textit{basis pursuit} and has been introduced by Chen and Donoho \cite{chen2001atomic}. 

What makes the basis pursuit method (\ref{eq:min_l1_norm_constrained_no_noise}) so interesting for estimating the parameter vector $\alpha$? Even if the number of measurements $y_1, ..., y_m$ is significantly smaller than the number of parameters $\alpha_1, ..., \alpha_p$, the solution to the basis pursuit problem (\ref{eq:min_l1_norm_constrained_no_noise}) provably equals the true parameter vector $\alpha$, provided that $\alpha$ is sufficiently sparse. To present one particular (selected) theorem on the performance of the basis pursuit method, we first introduce the \emph{restricted isometry property} (RP) of the measurement matrix $X$.
\begin{definition}[Restricted Isometry Property \cite{candes2008restricted}]\label{definition:restricted_isometry_property}
An $m \times p$ matrix $X$ satisfies the \emph{restricted isometry property} (RP) of order $s$ if there exists a $\delta_s \in (0, 1)$ such that
\begin{align*}
(1 - \delta_s)\lVert z \rVert^2_2 \le \lVert Xz \rVert^2_2 \le (1 + \delta_s)\lVert z \rVert^2_2,
\end{align*}
holds for all $s$-sparse vectors $z$.
\end{definition}
Hence, if the measurement matrix $X$ satisfies the RP of order $s$, then the length of any $s$-sparse vector $z$ stays almost (up to a relative deviation of $\delta_s$) constant under the linear map given by the matrix $X$. Theorem \ref{theorem:restricted_isometry_property} relates the RP of the measurement matrix $X$ to the solution $\hat{\alpha}$ of the basis pursuit (\ref{eq:min_l1_norm_constrained_no_noise}) and is due to Cand\`{e}s \cite{candes2008restricted}.

\begin{theorem}[]\label{theorem:restricted_isometry_property}
Suppose that the $m \times p$ measurement matrix $X$ satisfies the RP of order $2s$ with $\delta_{2s} <\sqrt{2} - 1$. Then, the solution $\hat{\alpha}$ to (\ref{eq:min_l1_norm_constrained_no_noise}) satisfies
\begin{align} \label{theorem:rip_bound}
\lVert \hat{\alpha} - \alpha \rVert_2 \le c_0 s^{-1/2} \lVert \alpha - \alpha^{(s)} \rVert_1
\end{align}
for some constant $c_0$, and where the $s$-sparse vector $\alpha^{(s)}$ follows from the parameter vector $\alpha$ by setting all but the largest $s$ components of $\alpha$ to zero.
\end{theorem}
Theorem \ref{theorem:restricted_isometry_property} is remarkable: If the parameter vector $\alpha$ is $s$-sparse, then it holds that $\alpha = \alpha^{(s)}$ and the estimation by (\ref{eq:min_l1_norm_constrained_no_noise}) is exact. Furthermore, Theorem \ref{theorem:restricted_isometry_property} implies the so-called \textit{oracle property}: Even if we knew (e.g. an oracle told us) which of the $p$ components of the parameter vector $\alpha$ are the $s$ largest, we could not improve the error bound (\ref{theorem:rip_bound}) by more than a constant value, and we refer to \cite{van2009conditions} for more details.

Thus, the RP\footnote{Similar results to Theorem \ref{theorem:restricted_isometry_property} can be derived based on the \textit{null space property} or the \textit{coherence} of the measurement matrix $X$, and both properties are related to the RP.}, and not the rank, of the measurement matrix $X$ guarantees the exact estimation of a sparse parameter vector $\alpha$. There are similarly elegant variations \cite{davenport2011introduction} of Theorem \ref{theorem:restricted_isometry_property} which rely on the RP of the measurement matrix $X$. The downside of the RP is that certifying whether a given matrix $X$ obeys the RP for a given order $2s$ and parameter $\delta_{s}$ is NP-hard \cite{bandeira2013certifying}. However, if the measurement matrix $X$ is generated at random, then the matrix $X$ satisfies the RP \emph{with very high probability} for various probability distributions \cite{candes2008introduction}, provided that the number of measurements $m$ satisfies
 \begin{align*}
 m \ge c s \log\left(\frac{p}{s}\right)
 \end{align*}
for some constant $c$. Hence, the number $m$ of observations grows only logarithmically with respect to the number $p$ of parameters $\alpha_1, ...,\alpha_p$ to recover an $s$-sparse parameter vector $\alpha$ exactly (with a high probability) if the measurement matrix $X$ is generated according to certain distributions.

If the measurements are subject to an additive error vector $e$, then the linear measurements (\ref{eq:lin_sys_cs}) become $y = X \alpha + e$. For the error-corrupted linear system $y = X \alpha + e$, there are variations of the basis pursuit optimisation problem (\ref{eq:min_l1_norm_constrained_no_noise}) with similar guarantees of an accurate recovery of the parameter vector $\alpha$. Most prominently, the least absolute shrinkage and selection operator (lasso) \cite{tibshirani1996regression}, or alternatively basis pursuit denoising, has found a wide application in statistics. In summary, a sparse signal can be obtained from significantly fewer samples than dictated by the Nyquist-Shannon sampling theorem, which gave rise to the whole field of \textit{compressed sensing}: Instead of acquiring a signal of tremendous dimension (e.g., a high resolution picture) and strongly compressing the signal in a subsequent step, compressed sensing aims to directly acquire the signal in compressed form. We refer the reader to \cite{hastie2015statistical, davenport2011introduction, candes2008introduction} for further details and theoretical results on sparse parameter estimation and related topics.

 \section{The Viral State Dynamics}
 \label{sec:viral_dynamics}  
When the viral state $v_i$ of a node $i$ corresponds to a group of individuals, it is a realistic assumption that the initial viral state $v_i[0]$, or the fraction of infected individuals of group $i$, is close to the healthy state. (For instance, consider the spread of a novel trend or tweet). We show below that if the initial viral state $v_i[1]$ of every node $i$ is sufficiently close to zero and $\tau > \tau^{(1)}_c$, then it holds $v_i[1] \le v_{\infty, i}$ for every node $i$, where $v_{\infty, i}$ is the steady-state of node $i$. 
  
If the effective infection rate satisfies $\tau < \frac{1}{\lambda_1}$, then the zero viral state is the only equilibrium \cite{van2014performance}, which is stable. If the effective infection rate satisfies $\tau \ge \frac{1}{\lambda_1}$, then the discrete-time NIMFA model (\ref{NIMFA_disc_stacked}) has an unstable equilibrium at the zero viral state and a stable non-zero equilibrium $v_\infty$. For heterogeneous spreading parameters, i.e. when the infection rate between two nodes $i$ and $j$ is $\beta_{ij} a_{ij}$ instead of $\beta a_{ij}$, similar results on the equilibria were derived \cite{pare2018analysis}. If the viral state $v[k]$ is close to zero, then the NIMFA equations (\ref{NIMFA_disc_stacked}) can be linearised around zero:
\begin{align}\label{eq:LTI_sys_NIMFA_zero}
v [k + 1] \approx \left((1 - \delta_T) I_N + \beta_T A\right) v[k] \quad \text{if} \quad v[k] \approx 0
\end{align}
In many applications, it is reasonable that most components of the initial viral state $v[0]$ are zero and only a few components are positive and small, and, hence, the approximation of the NIMFA equations (\ref{NIMFA_disc_stacked}) by the linear time-invariant (LTI) system (\ref{eq:LTI_sys_NIMFA_zero}) is accurate for small times $k$.

 The discrete-time NIMFA equations (\ref{NIMFA_disc_stacked}) can be transformed to an equivalent discrete-time system whose state at time $k$ is given by the difference of the viral state $v[k]$ to the steady state $v_\infty$, which we denote by $\Delta v[k] =v[k] - v_\infty$. 
\begin{proposition}[NIMFA Equations as Difference to the Steady-State]\label{proposition:delta_v_equation}
If $\tau \ge \frac{1}{\lambda_1}$, then the difference $\Delta v[k]$ from the viral state $v[k]$ to the steady state $v_\infty$ of the discrete-time NIMFA model (\ref{NIMFA_disc_stacked}) evolves according to
\begin{align} \label{delta_v_system}
\Delta v[k+1] & = F \Delta v[k]-\beta_T   \textup{\textrm{diag}}(\Delta v[k] )  A  \Delta v[k]   
\end{align}
where the $N\times N$ matrix $F$ is given by
\begin{align*}
F = I_N + \textup{\textrm{diag}}\left( \frac{\delta_T}{ v_{\infty, 1} - 1}, ...,  \frac{\delta_T}{ v_{\infty, N} - 1}\right) + \beta_T  \textup{\textrm{diag}}(u- v_{\infty} ) A
\end{align*}
\end{proposition}
\begin{proof}
Appendix \ref{appendix:difference_eq}.
\end{proof} 

For $k \rightarrow \infty$, the viral state $v[k]$ converges to the steady-state $v_\infty$ (if $\tau \ge \frac{1}{\lambda_1}$), which implies $\Delta v[k] \rightarrow 0$ for $k \rightarrow \infty$. We obtain 
\begin{align} \label{eq:LTI_sys_NIMFA_ss}
\Delta v[k+1] \approx F \Delta v[k]
\end{align}
for large $k$ by linearising (\ref{delta_v_system}) around the stable equilibrium $\Delta v[k] = 0$. Thus, for applications in which the initial viral state $v[0]$ is close to zero, we can describe the NIMFA equations by LTI systems in two different regimes: For either small or large times $k$, the NIMFA equations (\ref{NIMFA_disc_stacked}) can be approximated by the LTI systems (\ref{eq:LTI_sys_NIMFA_zero}) and (\ref{eq:LTI_sys_NIMFA_ss}), respectively. Furthermore, we obtain the following corollary from Proposition \ref{proposition:delta_v_equation}:
\begin{corollary}[]\label{corollary:below_steady_state}
If the sampling time $T$ satisfies
\begin{align}\label{eq:bound_sampling_time}
T \le \frac{1}{\delta  +\beta  d_i} 
\end{align}
for every node $i$, then it holds
\begin{align}\label{eq_below_steady_state}
v_i[0] \le v_{\infty, i} ~ \forall i \in \mathbb{N}_N \Rightarrow v_i[k] \le v_{\infty, i} ~ \forall k \ge 1, i \in \mathbb{N}_N 
\end{align}
and, analogously,
\begin{align*}
v_i[0] \ge v_{\infty, i} ~ \forall i \in \mathbb{N}_N \Rightarrow v_i[k] \ge v_{\infty, i} ~ \forall k \ge 1, i \in \mathbb{N}_N 
\end{align*}
\end{corollary}
\begin{proof}
Appendix \ref{appendix:below_steady_state}.
\end{proof} 

In other words, the set $\mathcal{V} = \{v | v_i \le v_{\infty, i}, ~\forall i \in \mathbb{N}_N\}$ is a \textit{positive invariant set} \cite{khalil1996nonlinear} of the NIMFA equations (\ref{NIMFA_disc_stacked}): If the viral state $v[0]$ is element of the set $\mathcal{V}$, then the viral state $v[k]$ will remain in the set $\mathcal{V}$ for $k \ge 0$. It is plausible that in most real-world spreading processes the viral state $v_i$ does not overshoot the steady-state $v_{\infty, i}$,  and Corollary \ref{corollary:below_steady_state} shows that the NIMFA model captures the realistic property of an absence of overshooting the steady-state $v_{\infty, i}$. We emphasise that Corollary \ref{corollary:below_steady_state} does not imply that the viral state $v[k]$ increases monotonically, as illustrated by Figure \ref{fig_not_monotonuous}.

\begin{figure}[h!]
	 \includegraphics[width=\textwidth]{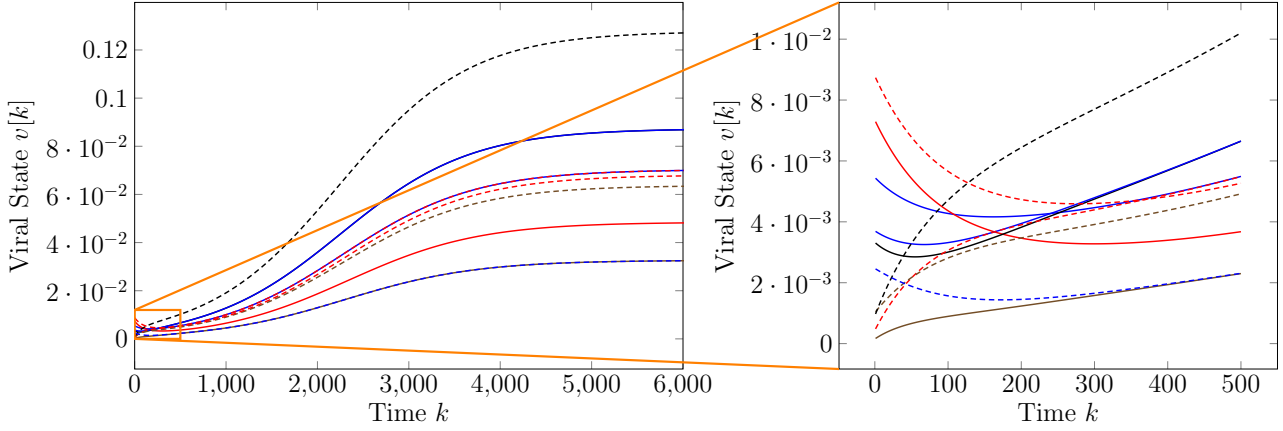}
	{\footnotesize\caption{The viral state traces $v[k]$ for a network with $N=10$ nodes. The initial viral states $v_i[1]$ were generated at random, independently and uniformly in the interval $[0, 0.01]$. Since $v_i[1] \le v_{\infty, i}$ for every node $i$, the implication (\ref{eq_below_steady_state}) is applicable. However, the viral states $v_i[k]$ are not monotonically increasing for every node $i$, as depicted by the right sub-plot. \label{fig_not_monotonuous}}}
\end{figure}

\section{Network Reconstruction and Spreading Parameter Estimation}
\label{sec:network_reconstruction}
We aim to estimate the adjacency matrix $A$ and the spreading parameters $\beta_T$ and $\delta_T$ given that the viral states $v_i[k]$ for all nodes $i=1, ..., N$ have been observed for the time instants $k=0, ..., n-1$. For simplicity\footnote{If the viral state $v_i[k]$ of a node $i$ does equal one for some time $k$, then our network reconstruction method can be adapted by discarding the observations of the respective time instants.}, we assume that $v_i[k] \neq 1$ for all nodes $i$ and all discrete times points $k$. 

In Subsection \ref{subsec:lin_equations}, we formulate the dependencies between the observed viral states $v[0], ..., v[n-1]$ as a set of equations which are linear with respect to the weighted adjacency matrix $\beta_T A$ and the curing rate $\delta_T$. We reformulate the linear system of Subsection \ref{subsec:lin_equations} in a more efficient manner in Subsection \ref{subsec:red_size}. Finally, the network reconstruction method is introduced in Subsection \ref{subsec:reconstruction_algorithm}

\subsection{Formulation as Linear System}
\label{subsec:lin_equations}
We show that the network reconstruction problem can be stated as set of equations which are linear in the curing rate $\delta_T$ and the $L_\textrm{max} \times 1$  weighted link vector $w$, where $w_l = \beta_T$ and $w_l = 0$ denotes the presence or absence of a link $l = 1, ..., {L_\textup{\textrm{max}}}$, respectively. More precisely, the NIMFA network reconstruction problem is of the form (\ref{eq:lin_sys_cs}): The parameter vector $\alpha$ is replaced by the $(L_\textrm{max} +1) \times 1$ vector $(w^T, \beta_T)^T$, and the measurements $y$ are given by a transform of the viral state observations $v[0], ...,  v[n-1] \in \mathbb{R}^N$.

We define the $N^2 \times 1$ vector $\beta_T A_\textrm{vec} = \beta_T (A^T_1, ..., A^T_N)^T$, which is obtained by concatenating the columns of the weighted adjacency matrix $\beta_T A$. Since the adjacency matrix $A$ is symmetric (i.e., $a_{ij} = a_{ji}$), we can express the $N^2 \times 1$ vector $\beta_T A_\textrm{vec}$ by the $L_\textup{\textrm{max}} \times 1$ weighted link vector $w$: For instance, for a network with $N=2$ nodes, it holds that $\beta_T A_\textrm{vec} = \beta_T (a_{11}, a_{21}, a_{12}, a_{22})^T$ and since $a_{12}= a_{21}$ we can write
\begin{align*}
\beta_T \begin{pmatrix}
a_{11}\\
a_{21} \\
a_{12}\\
a_{22}
\end{pmatrix} = \beta_T \begin{pmatrix}
1 & 0 & 0 \\
0 & 1 & 0 \\
0 & 1 & 0 \\
0 & 0 & 1 
\end{pmatrix}\begin{pmatrix}
a_{11}\\
a_{12}\\
a_{22}
\end{pmatrix}
\end{align*}
Furthermore, the graph has no self-loops (i.e. $a_{ii}=0$) and we obtain
\begin{align*}
 \beta_T \begin{pmatrix}
a_{11}\\
a_{21} \\
a_{12}\\
a_{22}
\end{pmatrix} = \begin{pmatrix}
0 \\
1 \\
1\\
0 
\end{pmatrix} w,
\end{align*}
where the weighted link vector equals $w = \beta_T a_{12}$ for a network with $N=2$ nodes. In general, there exists an $N^2 \times L_\textup{\textrm{max}}$ matrix $\Pi_N$ with zero-one entries for any number of nodes $N$ such that
\begin{align} \label{pi_mat}
\beta_T A_\textrm{vec} = \Pi_N w
\end{align}

To formulate the network reconstruction problem as linear system, we define three $n \times N$ matrices as
\begin{align}
B = \begin{pmatrix}
b^T[0]\\
\vdots \\ 
b^T[n-1]
\end{pmatrix}, \quad C = \begin{pmatrix}
c^T[0]\\
\vdots \\
c^T[n-1]
\end{pmatrix}, \quad
V = \begin{pmatrix}
v^T[0] \\
\vdots \\
v^T[n-1]
\end{pmatrix},\label{matriecsBCV}
\end{align} 
where 
\begin{align} \label{eq:def_bik}
b_i[k] &= \frac{v_i [k + 1] - v_i[k] }{1 - v_i[k]}
\end{align}
and
\begin{align}\label{eq:def_cik}
c_i[k] &= \frac{v_i[k]}{ v_i[k] -1},
\end{align}
respectively, for any time $k=0, ..., n-1$. Finally, we obtain Lemma \ref{lemma:reconstruction_as_linear_system}.

\begin{lemma}[Network Reconstruction as Linear System] \label{lemma:reconstruction_as_linear_system}
The viral state observations $v[0], ..., v[n-1]$, which were generated by the discrete-time NIMFA model (\ref{NIMFA_disc_stacked}), give rise to a set of linear equations for the $L_\textup{\textrm{max}} \times 1$ weighted link-vector $w = \beta_T (a_{12}, ..., a_{1N}, a_{23}, ..., a_{N-1, N})^T$ and the curing rate $\delta_T$:
\begin{align} \label{eq:lin_sys}
B_\textup{\textrm{vec}}  & =  M \begin{pmatrix} 
w \\
\delta_T
\end{pmatrix},
\end{align}
where the $n N \times ({L_\textup{\textrm{max}}}+1)$ matrix $M$ is given by
\begin{align} \label{eq:matrix_M}
M =  \begin{pmatrix}
I_N \otimes V & C_\textup{\textrm{vec}} 
\end{pmatrix}
\begin{pmatrix} 
\Pi_N & 0 \\
0 & 1
\end{pmatrix},
\end{align}
and the $N^2 \times {L_\textup{\textrm{max}}}$ matrix $\Pi_N$ is defined by (\ref{pi_mat}).
\end{lemma}
\begin{proof}
 Appendix \ref{appendix:reconstruction_as_linear_system}.
\end{proof}
Since $\textrm{rank}(I_N \otimes V) = N\textrm{rank}(V)$ and $\textrm{rank}(AB) \le \textrm{min} \{\textrm{rank}(A), \textrm{rank}(B)\}$ for any two matrices $A, B$, we obtain an upper bound on the rank of the matrix $M$
\begin{align}\label{rank_M_upperBound}
\textrm{rank}(M) \le N\textrm{rank}(V) +1, 
\end{align}
 which is crucial to solving the linear system (\ref{eq:lin_sys}).

Subject to the non-convex constraints $w_l\in \{0, \beta_T\}$ for all weighted links $l = 1, ..., L_\textup{\textrm{max}}$, the linear system (\ref{eq:lin_sys}) may have a unique solution for $w$ and $\delta_T$ even if the matrix $M$ does not have full column rank, i.e. $\textrm{rank}(M) < L_\textrm{max}+1$. \textit{Nevertheless, solving (\ref{eq:lin_sys}) subject to $w_l\in \{0, \beta_T\}$ for all links $l = 1, ..., L_\textup{\textrm{max}}$ seems to be computationally infeasible, and, hence, we omit the non-convex constraints $w_l\in \{0, \beta_T\}$ in the following}. If the constraints $w_l\in \{0, \beta_T\}$ are omitted, then the set of linear equations (\ref{eq:lin_sys}) has a unique solution for the weighted link-vector $w$ and the infection rate $\delta_T$ only if the matrix $M$ has full column rank, i.e. $\textrm{rank}(M) = {L_\textup{\textrm{max}}}+1$. The bound (\ref{rank_M_upperBound}) implies that the unconstrained set of linear equations (\ref{eq:lin_sys}) does not have a unique solution if the rank of the viral state matrix $V$ is smaller than $L_\textrm{max}/N = (N-1)/2$.

To solve a set of linear equations in practice (i.e., with finite precision arithmetic), the condition number of the respective matrix is decisive \cite{golub2012matrix}. With respect to the $\ell_2$-norm, the condition number $\kappa(V)$ of the viral state matrix $V$ equals $\kappa(V) = \sigma_1(V)/\sigma_\textrm{min}(V)$, where $\sigma_1$ and $\sigma_\textrm{min}$ denote the largest and smallest singular value of the matrix $V$, respectively. Closely related to the condition number $\kappa(V)$ is the \textit{numerical} rank $r$ of the viral state matrix $V$, which is the number of singular values of $V$ that are greater than a small threshold $\epsilon_\textrm{rank}$. By setting the threshold to $\epsilon_\textrm{rank} = \gamma \sigma_1(V)$ for some proportionality constant $\gamma >0$, the ratio of the largest to the smallest non-zero singular value of the truncated singular value decomposition $\tilde{V}_r$ of the matrix $V$ equals at most $1/\gamma$.
 
The average numerical rank of the viral state matrix $V$ is obntained by numerical simulations. We generate 100 instances of a Barab\'asi-Albert random graph \cite{barabasi1999emergence} with $N = 10$ to $N = 200$ nodes, where the initial number of nodes is set to $m_0 = 10$ and the number of links per addition of a new node is set to $m=3$. We set the effective infection rate to $\beta/\delta= 1.1 \tau^{(1)}_c$. For each network, the initial viral states $v_i[0]$ were generated at random, independently and uniformly in the interval $[0, 0.01 v_{\infty, i}]$. The observation length $n$ was set such that the viral state $v[k]$ (practically) converged to the steady-state $v_\infty$. Figure \ref{fig_rank_vs_N} shows that the numerical rank of the viral state matrix $V$ hardly changes with respect to the number of nodes $N$ when $N$ is greater than twenty. Even worse, the singular values $\sigma_1, \sigma_2, ...$ of the viral state matrix $V$ for networks of fixed size $N=100$ decrease exponentially fast, as illustrated by Figure \ref{fig_svd_plot}. Hence, the linear system (\ref{eq:lin_sys}) is very ill-conditioned, and it is almost impossible in practice to obtain the adjacency matrix $A$ and the spreading parameters $\beta_T, \delta_T$ only from solving the linear system (\ref{eq:lin_sys}).

\begin{figure}[h!]
	 \includegraphics[width=\textwidth]{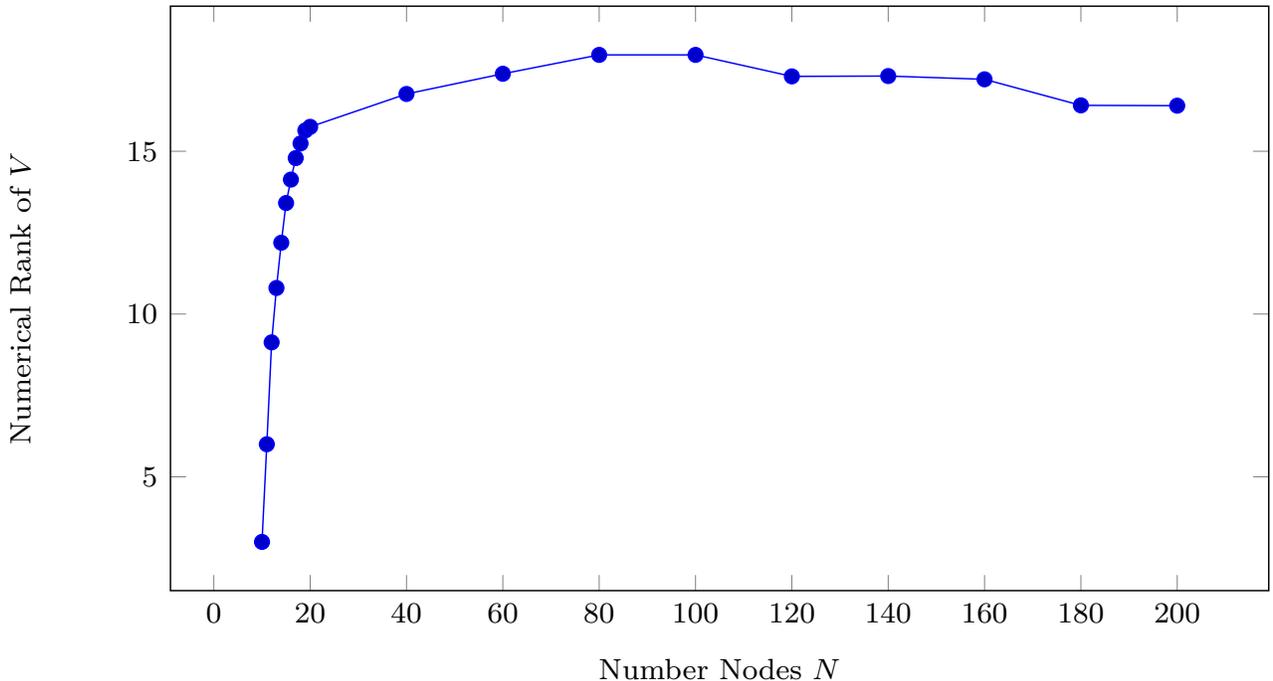}
	{\footnotesize\caption{The average numerical rank, implemented by the Matlab command \texttt{rank}, of the viral state sequence matrix $V$ versus the number of nodes $N$ for Barab\'asi-Albert random graphs. The results are averaged over 100 networks which were generated by the Barab\'asi-Albert random graph model.\label{fig_rank_vs_N}}}
\end{figure}

\begin{figure}[h!]
	 \includegraphics[width=\textwidth]{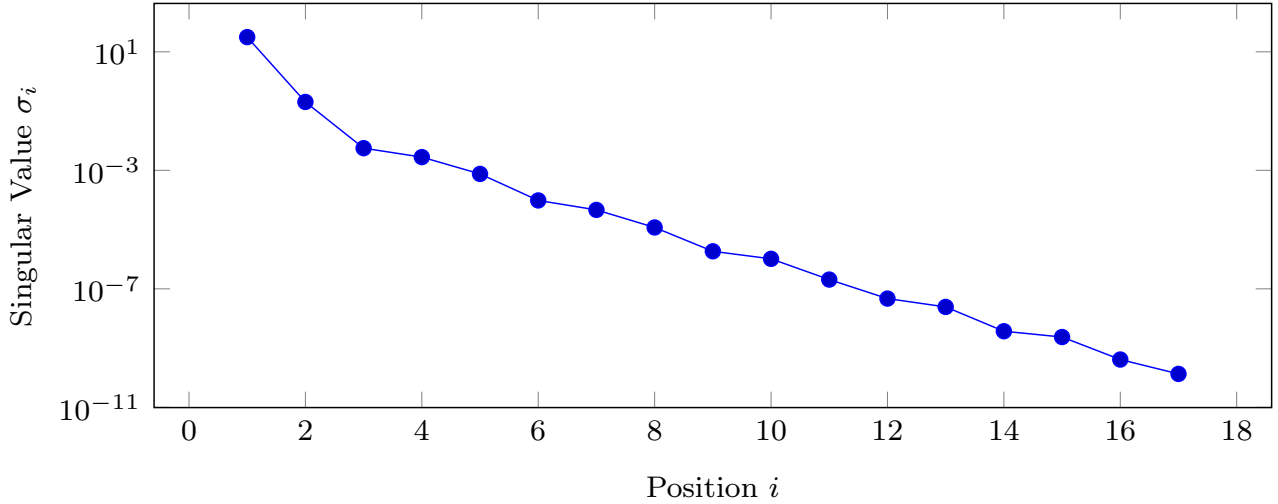}
	{\footnotesize\caption{The singular values $\sigma_i$ of the viral state matrix $V$ in descending order. Only the singular values above the threshold set by the Matlab command \texttt{rank} are considered. The results are averaged over 100 networks with $N=100$ nodes which were generated by the Barab\'asi-Albert random graph model. \label{fig_svd_plot}}}
\end{figure}

\subsection{Reduced-Size Linear System}
\label{subsec:red_size}
Since the observation length $n$ is considerably greater than the numerical rank of the $n \times N$ viral state matrix $V$, we resort to the truncated singular value decomposition (TSVD) to store and process the matrix $V$ in an efficient manner. We denote the numerical rank of the viral state matrix $V$ by $r$, and we replace the matrix $V$ in (\ref{eq:matrix_M}) by its TSVD $\tilde{V}_r = U_r S_r Q^T_r$, which is obtained by considering only the largest $r$ singular values and setting the other singular values to zero.

Before stating the main result of this section, we introduce two $r \times N$ matrices $\tilde{B}$ and $\tilde{C}$ whose columns $\tilde{B}_l, \tilde{C}_l \in \mathbb{R}^r$ are given by
\begin{align*}
\tilde{B}_l = \left(U_r S_r \right)^\dagger B_l 
\end{align*}
and
\begin{align*}
\tilde{C}_l = \left(U_r S_r \right)^\dagger C_l
\end{align*}
for every $l=1,..., N$, where $\left(U_r S_r \right)^\dagger$ denotes the Moore-Penrose pseudoinverse of the matrix $U_r S_r$.
\begin{lemma}[Reduced-Size Linear System] \label{lemma:reduced_size_lin_sys}
When the $n \times N$ viral state matrix $V$ is replaced by its TSVD $\tilde{V}_r$ with $r$ positive singular values, the linear system (\ref{eq:lin_sys}) can be approximated by
\begin{align} \label{eq:setlineqfinal_reduced}
\tilde{B}_{\textup{\textrm{vec}}}  &=   \tilde{M} \begin{pmatrix} 
w \\
\delta_T
\end{pmatrix}
\end{align}
where the $rN \times (L_\textup{\textrm{max}}+1)$ matrix $\tilde{M}$ is given by
\begin{align*}
 \tilde{M} =\begin{pmatrix}
I_N \otimes Q^T_r & \tilde{C}_\textup{\textrm{vec}} 
\end{pmatrix}
\begin{pmatrix} 
\Pi_N & 0 \\
0 & 1
\end{pmatrix}
\end{align*}
\end{lemma}
\begin{proof}
 Appendix \ref{appendix:reduced_size_lin_sys}.
\end{proof}
The reduced-size linear system (\ref{eq:setlineqfinal_reduced}) has $rN$ equations as compared to $nN$ equations in the full-size linear system (\ref{eq:lin_sys}). We emphasise that the matrix $\tilde{M}$ is sparse, and multiplying an $(L_\textrm{max}+1)\times 1$ vector by the matrix $\tilde{M}$ can be implemented efficiently.

Figure \ref{fig_rank_vs_N} indicates that the numerical rank of the viral state matrix $V$ does not increase with the number of nodes $N$. \textit{Hence, it may be necessary to observe multiple epidemic outbreaks to reconstruct the network accurately}. The linear system (\ref{eq:setlineqfinal_reduced}) can be extended straightforwardly for multiple epidemic outbreaks (or realisations): For the $l$-th epidemic outbreak, we denote the viral state at time $k$ by $v^{(l)}[k]$ and the resulting viral state matrix by $V^{(l)} = (v^{(l)}[0], ..., v^{(l)}[n_l-1])$, where $n_l$ is the number of observations. From each epidemic outbreak $l$ we obtain a linear system (\ref{eq:setlineqfinal_reduced}) with a matrix $\tilde{M}^{(l)}$ and a vector $\tilde{B}^{(l)}_\textrm{vec}$. We stack the equation (\ref{eq:setlineqfinal_reduced}) for every outbreak $l$ to finally obtain 
\begin{align} \label{eq:setlineqfinalfinal}
B_{\textrm{all}}  &=  M_{\textrm{all}} \begin{pmatrix} 
w \\
\delta_T
\end{pmatrix},
\end{align}
where the vector $B_{\textrm{all}}$ and the matrix $M_{\textrm{all}}$ are obtained by stacking the vectors $\tilde{B}^{(l)}_{\textrm{vec}}$ and the matrices $\tilde{M}^{(l)}$, respectively, for every epidemic outbreak $l$.

\subsection{Network Reconstruction Algorithm}
\label{subsec:reconstruction_algorithm}
Based on (\ref{eq:setlineqfinalfinal}), we formulate the network reconstruction as a \textit{basis pursuit} problem with non-negativity constraints:
\begin{align} \label{net_rec_basis_pursuit}
 \begin{aligned}&\underset{w, \delta_T}{\text{min}}  & & \left\lVert \begin{pmatrix} 
w \\
\delta_T
\end{pmatrix} \right\rVert_1 & \\
 &\text{s.t.} & & M_{\textrm{all}}  \begin{pmatrix} 
w \\
\delta_T
\end{pmatrix}  = B_{\textrm{all}} & \\ 
  & & & w_{l} \ge 0, \quad l = 1, ..., L_\textup{\textrm{max}} & \\
   & & &  \delta_T \ge 0& 
\end{aligned} 
 \end{align}

We pose the optimisation problem (\ref{net_rec_basis_pursuit}) only for the \textit{support recovery} of the weighted link-vector $w$, i.e. only for estimating which entries of the link-vector $w$ are non-zero. The solution to (\ref{net_rec_basis_pursuit}) for the weighted link-vector is denoted by $\hat{w}$ (the solution to (\ref{net_rec_basis_pursuit}) for the curing rate $\delta_T$ is discarded). The number of optimisation variables in (\ref{net_rec_basis_pursuit}) equals $L_\textup{\textrm{max}}+1$ and, hence, grows quadratically with the number of nodes $N$. To solve the optimisation problem (\ref{net_rec_basis_pursuit}) for large-scale networks, we adopt the alternating direction method of multipliers (ADMM) algorithm and make use of the Matlab implementation of Boyd \emph{et al.} \cite{boyd2011distributed}, which is available online. We apply the LSQR method \cite{paige1982lsqr} for the $x$-update in the ADMM algorithm, which is computationally beneficial since the matrix $M_\textrm{all}$ is sparse, and we rely on the Matlab command \texttt{lsqr}.

Numerically, a solution to (\ref{net_rec_basis_pursuit}) cannot be obtained with infinite precision, and even if the exact solution $\hat{w}_l$ for the $l$-th weighted link is zero, the respective numerical solution may be non-zero. We infer the presence of the link $l$ only if the estimate $\hat{w}_l$ exceeds a threshold $\epsilon$: 
\begin{align}\label{x_hat_bin}
\hat{x}_l(\epsilon) = \begin{cases}
1 \quad \text{if} \quad \hat{w}_l\ge \epsilon \\
0\quad \text{if} \quad \hat{w}_l< \epsilon
\end{cases}
\end{align}

It remains to set the value of the threshold $\epsilon$. As an example, Figure \ref{fig_link_rounding_plot} depicts a histogram of the numerical solutions for the weighted link estimates $\hat{w}_l$ for a network with $N=100$ nodes, for which the threshold $\epsilon$ can be set ``intuitively''. In general, we choose to set the threshold as follows:
\begin{align}\label{eps_lambda}
\epsilon^*= \textrm{ arg }\underset{\epsilon}{\textrm{min}} ~ \textrm{Var}\left[S_+(\epsilon)\right] + \textrm{Var}\left[S_-(\epsilon)\right],
\end{align}
where $S_+(\epsilon) = \left\{l \in \mathbb{N}_{L_\textup{\textrm{max}}}| \hat{w}_l \ge \epsilon \right\}$  denotes the set of weighted link estimates that are greater than a threshold $\epsilon$, and $\textrm{Var}\left[S_+(\epsilon)\right]$ denotes the population variance of the set $\{\hat{w}_l | l \in S_+(\epsilon)\}$. (The set $S_-(\epsilon)$ and the variance $\textrm{Var}\left[S_-(\epsilon)\right]$ are defined analogously, but with regard to weighted link estimates $\hat{w}_l$ that are \textit{smaller} than a threshold $\epsilon$.) Hence, the threshold $\epsilon^*$ in (\ref{eps_lambda}) minimises the variance of the weighted link estimates $\hat{w}_l$ that are left, and right, of the threshold $\epsilon^*$. In practice, the optimal threshold may not exactly equal to $\epsilon^*$, and we also consider values for the threshold $\epsilon$ that are close to $\epsilon^*$. More precisely, additionally to the threshold $\epsilon^*$, we consider 100 thresholds $\epsilon^-_1 < ... < \epsilon^-_{100} < \epsilon^*$ that are smaller than $\epsilon^*$ and 100 thresholds $\epsilon^+_1 > ... >\epsilon^+_{100}> \epsilon^*$ that are greater than $\epsilon^*$, where we choose the spacing of two thresholds $\epsilon^-_{i}$ and $\epsilon^-_{i+1}$ such that there is exactly one link $l$ whose estimate $\hat{w}_l$ satisfies $\epsilon^-_{i} < \hat{w}_l < \epsilon^-_{i+1}$ (and similarly for $\epsilon^+_{i}$ and $\epsilon^+_{i+1}$). We denote the set of all considered thresholds as $S_\epsilon = \{\epsilon^-_l| l\in \mathbb{N}_{100}\} \cup \{\epsilon^+_l| l\in \mathbb{N}_{100}\} \cup \{\epsilon^*\}$. As a result of the 201 different thresholds $\epsilon \in S_\epsilon$, we obtain 201 different estimates $\hat{x}(\epsilon)$ for the link vector $x$.

\begin{figure}[h!]
	 \includegraphics[width=\textwidth]{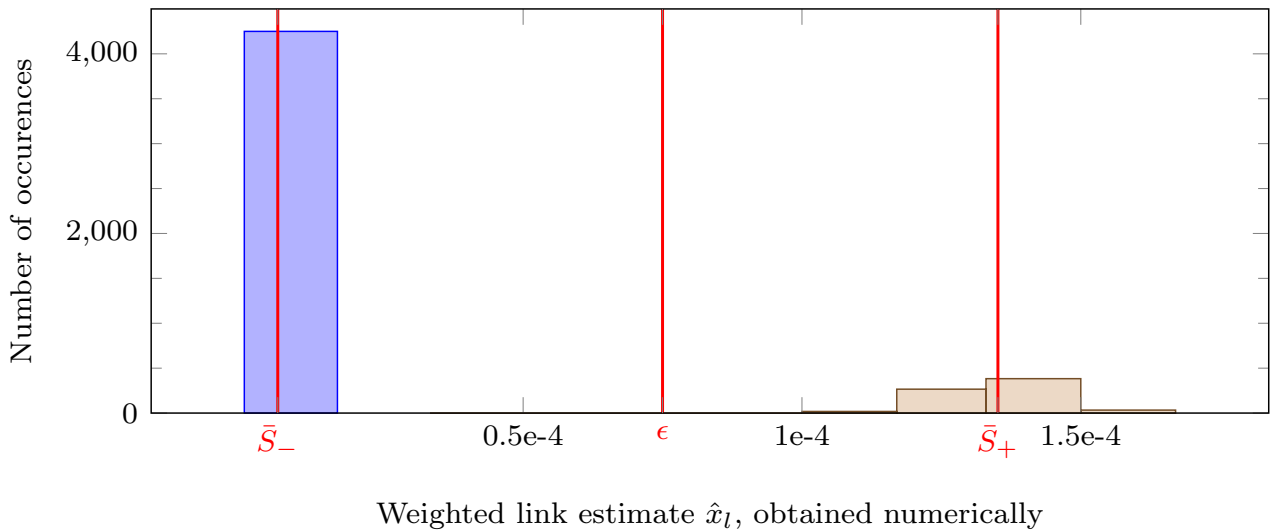}
	{\footnotesize\caption{ For a network with $N=100$ nodes, this is a histogram of the numerical solution $\hat{w}_l$ to (\ref{net_rec_basis_pursuit}) for every link $l= 1, ..., 4950$. Intuitively, one may set the threshold $\epsilon$ to $7.5 \cdot 10^{-5}$, which results in the blue bar left of the threshold $\epsilon$ and beige bars right the threshold $\epsilon$. The respective mean of the values left and right of the threshold $\epsilon$ are denoted by $\bar{S}_-$ and $\bar{S}_+$.
	 \label{fig_link_rounding_plot}}}
\end{figure}

After the support recovery by (\ref{net_rec_basis_pursuit}) and (\ref{x_hat_bin}), we aim for a refinement of the estimates for the spreading parameters $\delta_T$ and $\beta_T$. For every threshold $\epsilon \in S_\epsilon$, we pose the refinement step as a constrained linear least-squares problem with respect to the spreading parameters $\beta_T$ and $\delta_T$ only:
 \begin{align} \label{opt_prob_beta_delta}
 \begin{aligned}
(\hat{\beta}_T(\epsilon), \hat{\delta}_T(\epsilon)) = &\text{ arg }\underset{\beta_T, \delta_T}{\text{min}}  & & \left\lVert B_{\textrm{all}} - M_{\textrm{all}}  \begin{pmatrix} 
\hat{x}(\epsilon) \beta_T \\
\delta_T
\end{pmatrix} \right\rVert^2_2& \\
 &\text{ s.t.} & &  \beta_T  \ge 0& \\
   & & &  \delta_T \ge 0& 
\end{aligned} 
 \end{align} 
  The heuristic estimation procedures (\ref{x_hat_bin}) and (\ref{opt_prob_beta_delta}) are executed for each of the 201 thresholds $\epsilon$ in the set $S_\epsilon$.  Thus, we obtain 201 different estimates for the spreading parameters $\beta_T, \delta_T$ and the link-vector $x$. Every threshold $\epsilon \in S_\epsilon$ results in a mean squared error $\textrm{MSE}(\epsilon)$, which equals the optimal value of (\ref{opt_prob_beta_delta}). We denote the threshold in $S_\epsilon$ which resulted in the minimal error $\textrm{MSE}(\epsilon)$ as $\epsilon_\textrm{opt}$. The complete estimation procedure is given in pseudo-code by Algorithm \ref{ALG_CLASSO}.

\begin{algorithm}
\caption{Heuristic NIMFA Network Reconstruction}
\begin{algorithmic}[1]
\State \textbf{Input: } viral state matrix $V$
\State \textbf{Output: } estimates for the adjacency matrix $\hat{A}_\textrm{vec}$ and spreading parameters $\hat{\beta}_T, \hat{\delta}_T$
\State $\hat{w} \gets$ solution to (\ref{net_rec_basis_pursuit})
\State Determine $\epsilon^*$ and $S_\epsilon$ by (\ref{eps_lambda})
\For {$\epsilon \in S_\epsilon$}
\State $\hat{x}(\epsilon) \gets$ threshold elements of $\hat{w}$ by (\ref{x_hat_bin})
\State $(\hat{\beta}_T(\epsilon), \hat{\delta}_T(\epsilon)) \gets$ solution to (\ref{opt_prob_beta_delta})
\State $\textrm{MSE}(\epsilon) \gets$ minimum of (\ref{opt_prob_beta_delta})
\EndFor
\State $\epsilon_\textrm{opt} \gets$ minimiser of $\textrm{MSE}(\epsilon)$ for $\epsilon \in S_\epsilon$
\State $\hat{A}_\textrm{vec} \gets \Pi_N \hat{x} (\epsilon_\textrm{opt})$
\State $\hat{\beta}_T \gets \hat{\beta}_T(\epsilon_\textrm{opt})$
\State $\hat{\delta}_T \gets \hat{\delta}_T(\epsilon_\textrm{opt})$
\end{algorithmic}
\label{ALG_CLASSO}
\end{algorithm}

\section{Numerical Evaluation}
\label{sec:numerical_evaluation}
In the following, we study the performance of Algorithm \ref{ALG_CLASSO} numerically. The curing rate is set to $\delta=1$, and the effective infection rate is set to $\beta/\delta= 1.1 \tau^{(1)}_c$, where $\tau^{(1)}_c$ is the epidemic threshold. The inequality (\ref{eq:bound_sampling_time}) on the sampling time $T$ depends on the degree $d_i$ of every node $i$, but we would like to set the sampling time $T$ only in dependency of the network size $N$. The denominator $\delta + \beta d_i$ in (\ref{eq:bound_sampling_time}) is upper bounded by $\delta ( 1 + 1.1 \tau^{(1)}_c (N-1))$ since the degree $d_i$ of every node $i$ is bounded by $d_i\le N-1$ and the infection rate is set to $\beta= 1.1\delta \tau^{(1)}_c$. Furthermore, the epidemic threshold equals $\tau^{(1)}_c = 1/\lambda_1$, from which follows that $\tau^{(1)}_c (N-1) \le 2/N$ since $\lambda_1 \ge 2L/N$ and $L\ge N-1$ for a connected graph \cite{van2010graph}. Finally, by setting the sampling time to 
\begin{align*}
T = \frac{1}{\delta(1 +1.1 N/2)}
\end{align*}
we ensure that the inequality (\ref{eq:bound_sampling_time}) holds for every network of size $N$.

The initial viral states $v_i[0]$ were generated at random, independently and uniformly in the interval $[0, 0.01 v_{\infty, i}]$, and thus the condition (\ref{eq_below_steady_state}) holds. The observation length $n$ was set such that the viral state $v[k]$ (practically) converged to the steady-state $v_\infty$. The considered networks were generated by two different random graph models. First, we consider the Watts-Strogatz random graph model \cite{watts1998collective} with rewiring probability $p_\textrm{WS} = 0.2$ and a varying mean node degree $d_\textrm{WS}$ as specified further below. Second, we consider the Barab\'asi-Albert random graph model \cite{barabasi1999emergence} with sparsity parameters $m_0 = 10$ and $m$ as specified further below.

We define the estimation error of the adjacency matrix as the sum of erroneous links, normalised by the total number of links in the true network:
\begin{align*}
\varepsilon_A = \frac{\sum^N_{i = 1}\sum^N_{j = 1} |a_{ij} - \hat{a}_{ij}|}{\sum^N_{i = 1}\sum^N_{j = 1} a_{ij}} 
\end{align*}
Here, $\hat{a}_{ij}$ denote the elements of the estimated adjacency matrix $\hat{A}$. The error $\varepsilon_\beta$ of the estimate $\hat{\beta}$ of the infection rate is defined as the relative deviation $|\beta-\hat{\beta}|/\beta$, and the error $\varepsilon_\delta$ of the estimate of the curing rate $\hat{\delta}$ is defined analogously.

Given the viral state sequence $v[0], ..., v[k]$ at time $k$, Algorithm \ref{ALG_CLASSO} returns estimates $\hat{\beta}, \hat{\delta}, \hat{A}$ for the spreading parameter and the adjacency matrix. By replacing the true adjacency matrix $A$ and the true spreading parameters $\beta, \delta$ in the NIMFA equations (\ref{NIMFA_disc_stacked}) by their respective estimates and setting the initial state at time $k$ to $v[k]$, we obtain a prediction $\hat{v}[k+1], ..., \hat{v}[n]$ of the viral state. We define the error of the viral state prediction as the averaged relative deviation of the viral states:
\begin{align*}
\varepsilon_v[k] = \frac{1}{n - k} \sum^{n}_{l = k+1} \frac{1}{N} \sum^{N}_{i=1} \frac{\left| v_i[l] - \hat{v}_i[l]\right|}{\left| v_i[l]\right|}
\end{align*}

\subsection{Impact of Sparsity on the Reconstruction Accuracy}
\label{subsec:sparsity} 
We would like to know when the network can be reconstructed accurately from only one epidemic outbreak. Motivated by the theoretical results in Section \ref{sec:lasso}, we focus on studying the impact of the network sparsity on the estimation performance of Algorithm \ref{ALG_CLASSO}. Figure \ref{fig_sparsity_WS} and Figure \ref{fig_sparsity_BA} illustrate that the network sparsity, rather than the number of nodes $N$, is decisive for an accurate reconstruction of the network from a single epidemic outbreak. Furthermore, the curing rate $\delta_T$ can be estimated accurately even if the network reconstruction error $\varepsilon_A$ is great.

\begin{figure}[h!]
	 \includegraphics[width=\textwidth]{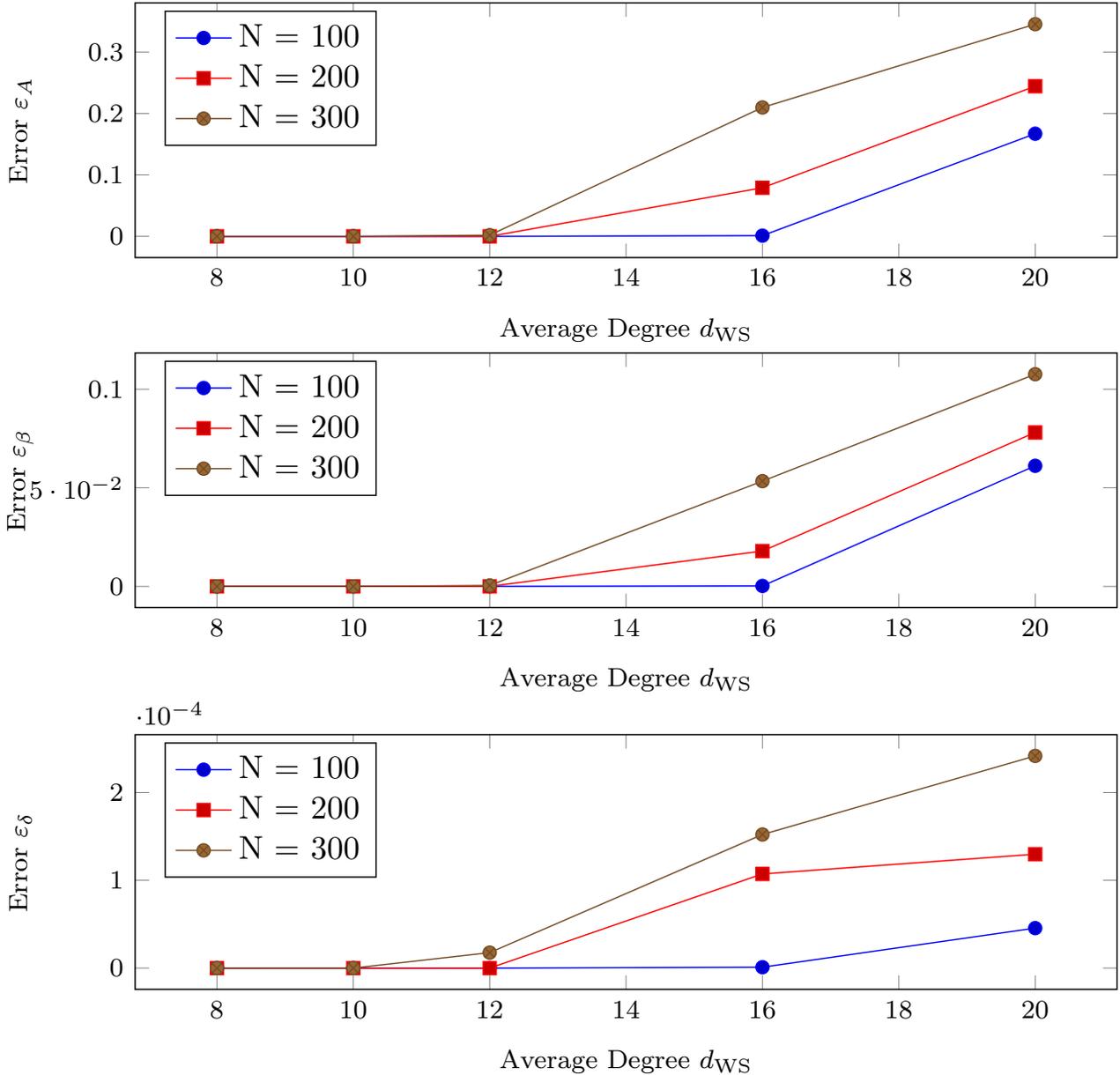}
	{\footnotesize\caption{ The estimation accuracy of Algorithm \ref{ALG_CLASSO} versus the average degree $d_\textrm{WS}$, averaged over 50 networks which were generated by the Watts-Strogatz random graph model. \label{fig_sparsity_WS}}}
\end{figure}

\begin{figure}[h!]
	 \includegraphics[width=\textwidth]{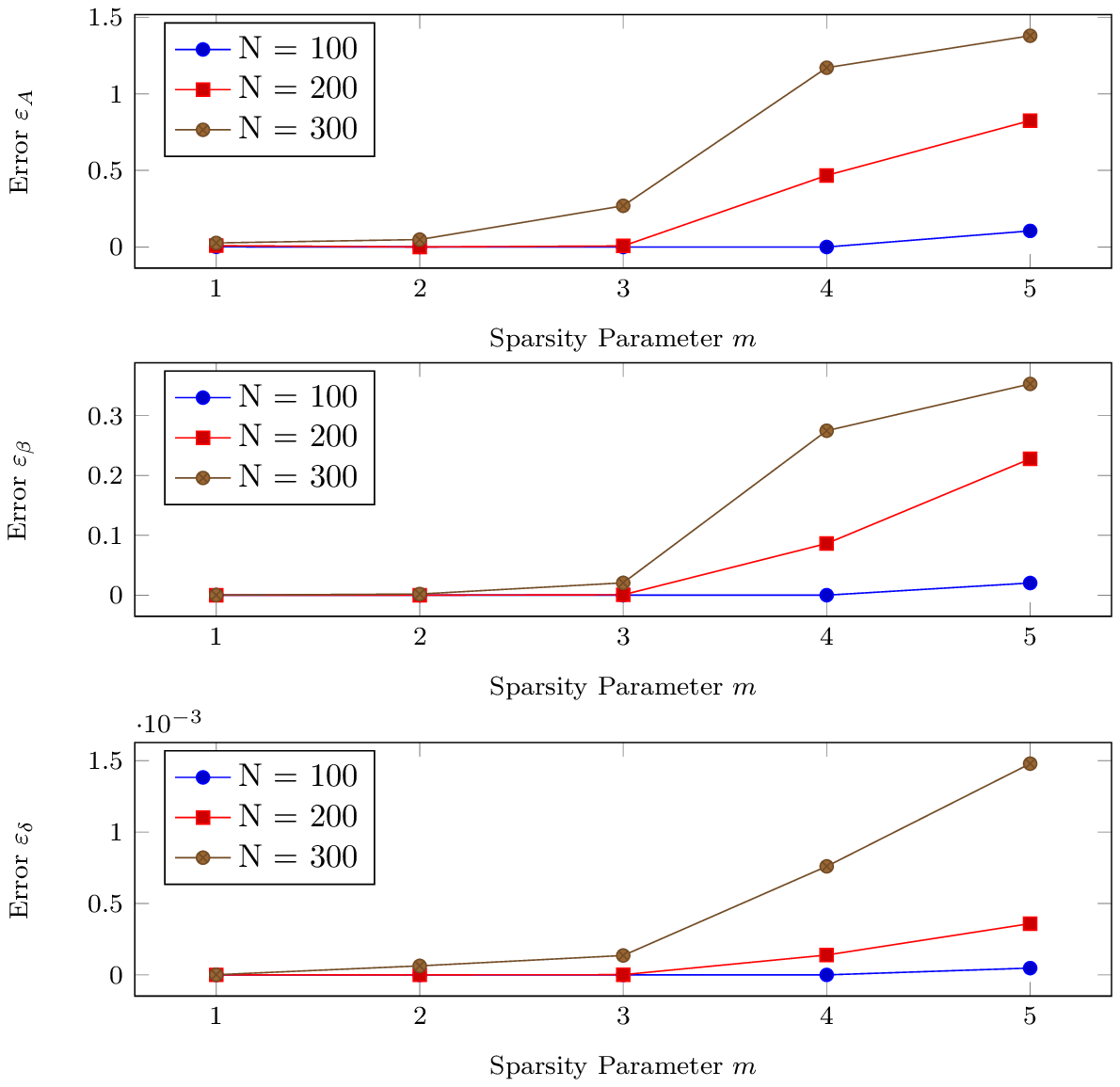}
	{\footnotesize\caption{ The estimation accuracy of Algorithm \ref{ALG_CLASSO} versus the sparsity parameter $m$, averaged over 50 networks which were generated by the Barab\'asi-Albert random graph model. \label{fig_sparsity_BA}}}
\end{figure}

If the network is not sparse, the basis pursuit method does not yield an accurate estimate of the network, and it is necessary to observe more than one epidemic outbreak to reconstruct the network accurately. Figure \ref{fig_dense_WS} shows that the accuracy of the network reconstruction increases quickly when more epidemic outbreaks are observed.

\begin{figure}[h!]
	 \includegraphics[width=\textwidth]{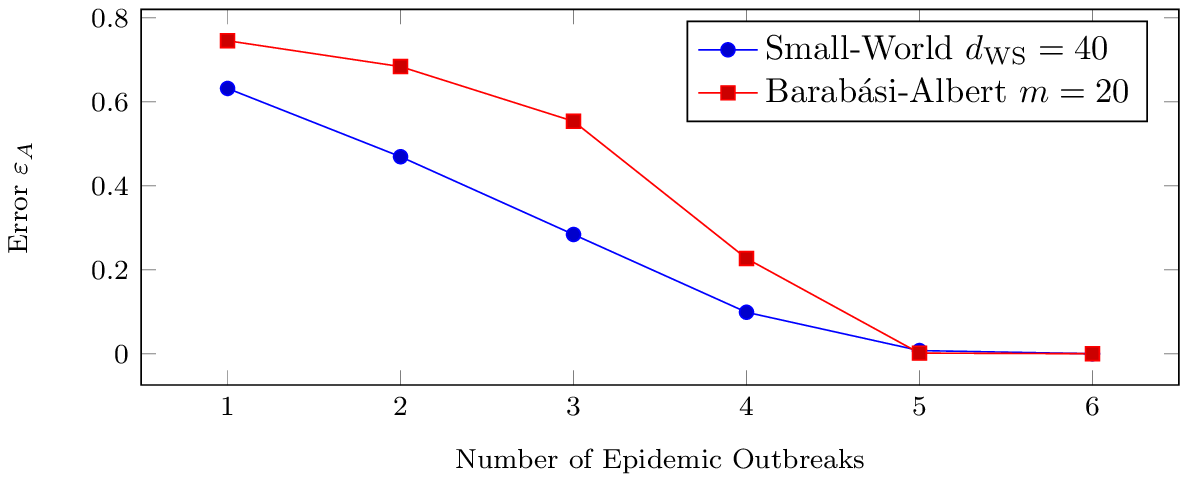}
	{\footnotesize\caption{ The estimation accuracy of Algorithm \ref{ALG_CLASSO} versus the number of epidemic outbreaks, averaged over 100 networks which were generated by the Watts-Strogatz and Barab\'asi-Albert random graph model with the parameters set to $d_\textrm{WS} = 40$ and $m= 20$, $m_0= 20$, respectively. \label{fig_dense_WS}}}
\end{figure}

\subsection{Predicting Epidemic Outbreaks}
Of particular practical interest is the prediction of an epidemic outbreak \textit{before} the steady state $v_\infty$ is reached: On the one hand, the prediction of the spread of an infectious disease allows for sophisticated countermeasures, such as the vaccination of groups with the greatest risk of infection. On the other hand, forecasting the evolution of a trend, the choice of a consumer product or a tweet among individuals enables targeted strategies aiming for \textit{increasing} the prevalence of the spread.

As discussed in Subsection \ref{subsec:sparsity}, sufficiently sparse networks can be reconstructed accurately from only one epidemic outbreak. Furthermore, it is even possible to reconstruct sparse networks by observing only the initial stage of the epidemic outbreak: We generate 100 networks of size $N=100$ by the Watts-Strogatz and the Barab\'asi-Albert random graph models with an average degree of $d_\textrm{WS} = 8$ and sparsity parameter $m = 1$, respectively. For each of those randomly generated networks, we run the network reconstruction method of Algorithm \ref{ALG_CLASSO} based on a viral state trace $v[0], ..., v[n-1]$ of observation length $n$ for increasing values of $n$. 

We define $n_\textrm{max}$ as the number of viral state observations such that the viral state can be predicting with a 95\% accuracy, more precisely: $\varepsilon_v[n] \le 0.05$ for all $n \ge n_\textrm{max}$ and \textit{for all} of the 100 randomly generated networks. For sparse networks, the epidemic outbreak can be predicted accurately, well before the steady state $v_\infty$ is reached as illustrated by Figure \ref{fig_prediction_mBA1_kWS4}. 

\begin{figure}[h!]
	 \includegraphics[width=\textwidth]{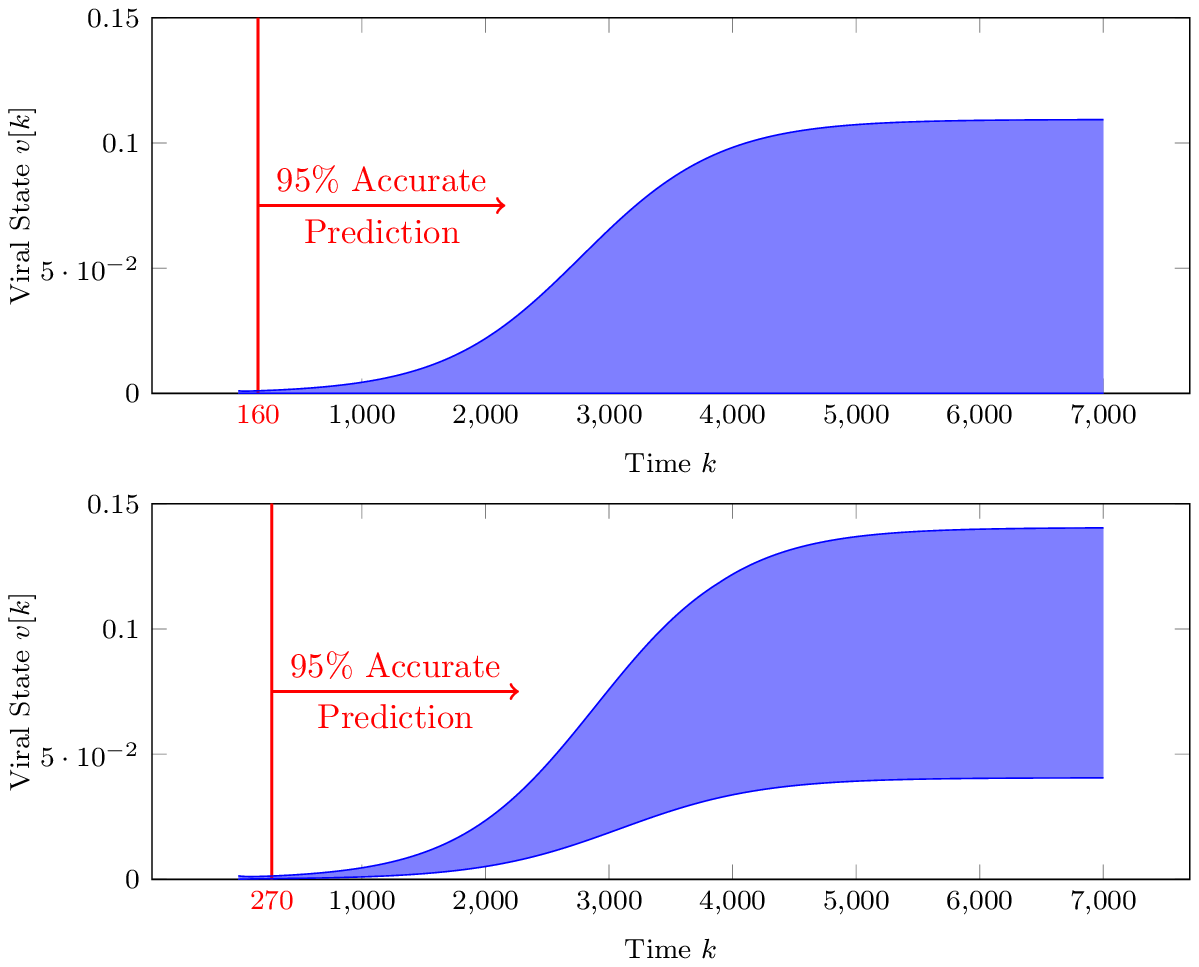}
	{\footnotesize\caption{The upper plot refers to Barab\'asi-Albert random networks with sparsity parameter $m=1$, and the lower plot refers to Watts-Strogatz random graphs with an average degree of $d_\textrm{WS} = 8$. Both plots refer to networks of size $N=100$. The viral state trace $v_i[0], ..., v_i[7000]$ for every node $i$ and for each of the 100 randomly generated networks falls in the respective blue envelope. The red line denotes the maximum number of observations $n_\textrm{max}$ which are required for an 95\% accurate viral state prediction ($\varepsilon_v[n_\textrm{max}] \le 0.05$) for \textit{each} of the 100 networks. \label{fig_prediction_mBA1_kWS4}}}
\end{figure}

\subsection{Computation Time}
Given a single epidemic outbreak, the computation time $T_\textrm{comp}$ of Algorithm \ref{ALG_CLASSO} versus the network size $N$ for Watts-Strogatz and Barab\'asi-Albert networks is depicted on a log-log scale in Figure \ref{fig_copm_time_WS} and Figure \ref{fig_copm_time_BA}, respectively. For Watts-Strogatz networks, the curves in Figure \ref{fig_copm_time_WS} indicate that the computation time $T_\textrm{comp}$ grows polynomially with respect to the number of nodes $N$: By fitting the function $\log_{10}(T_\textrm{comp}) = \theta_1\log_{10}(N) +  \theta_2$ to the data points in Figure \ref{fig_copm_time_WS}, we obtain the values for the parameters $\theta_1, \theta_2$ given in Table \ref{table:parameters_fit}. Hence, the computation time $T_\textrm{comp}$ grows with $\mathcal{O}(N^{\theta_1})$, where $\theta_1 \approx 3.5$. The average degree $d_\textrm{WS}$ (or, the sparsity) of the network does not have a significant impact on the computation time $T_\textrm{comp}$. For Barab\'asi-Albert networks, Figure \ref{fig_copm_time_BA} indicates that the computation time $T_\textrm{comp}$ does not follow as clearly as for Watts-Strogatz networks. However, the slope of the curve in Figure \ref{fig_copm_time_BA} saturates for $N\ge 200$ at a value of approximately $\theta \approx 1.7$. Thus, the computation time $T_\textrm{comp}$ with respect to the number of nodes $N$ for Barab\'asi-Albert networks can be upper bounded by a polynomial. A computation time that is polynomial with the network size $N$ renders the network reconstruction approach of Algorithm \ref{ALG_CLASSO} feasible in practice.

\begin{center}
  \begin{table}
  \centering
  \begin{tabular}{ | l | l | l | l |}\hline
    $d_\textrm{WS}$ & 8 & 10 & 12 \\ \hline \hline 
     $\theta_1$ & 3.3996  & 3.6343 & 3.4479 \\ \hline 
    $\theta_2$ & -11.5286 & -12.2139  & -10.5584 \\ \hline 
   \end{tabular}  
   \caption{Parameters obtained by fitting polynomial functions to the the computation time $T_\textrm{comp}$ in dependency of the number of nodes $N$. \label{table:parameters_fit}}   
   \end{table}
\end{center}

\begin{figure}[h!]
	 \includegraphics[width=\textwidth]{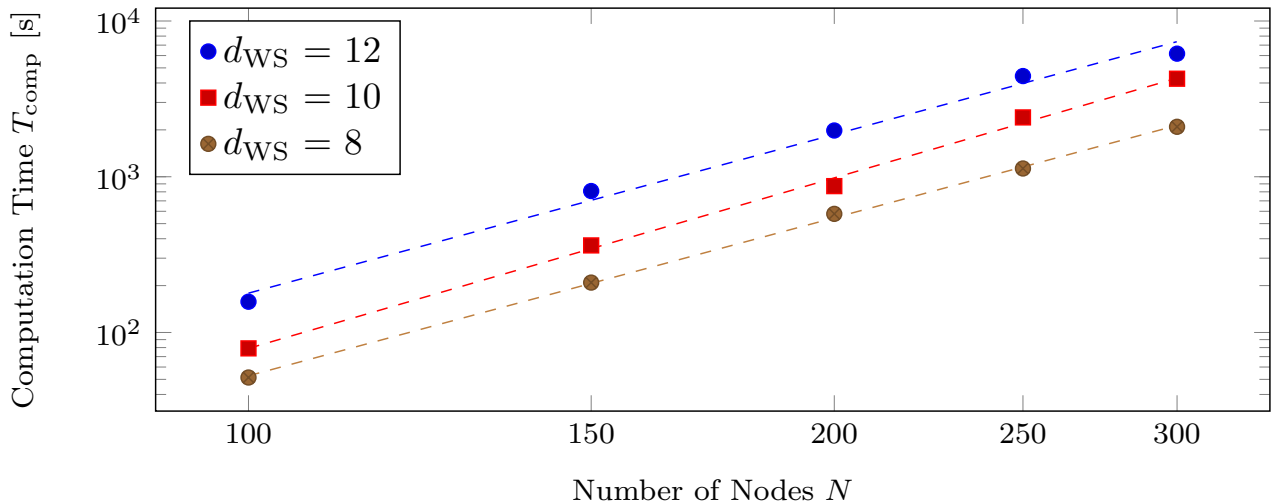}
	{\footnotesize\caption{ The computation time $T_\textrm{comp}$ of Algorithm \ref{ALG_CLASSO} versus the number of nodes $N$ on a log-log scale, averaged over 50 networks which were generated by the Watts-Strogatz random graph model with average degree $d_\textrm{WS} = 8, 10, 12$. The fitted linear functions are given by the dashed curves. \label{fig_copm_time_WS}}}
\end{figure}

\begin{figure}[h!]
	 \includegraphics[width=\textwidth]{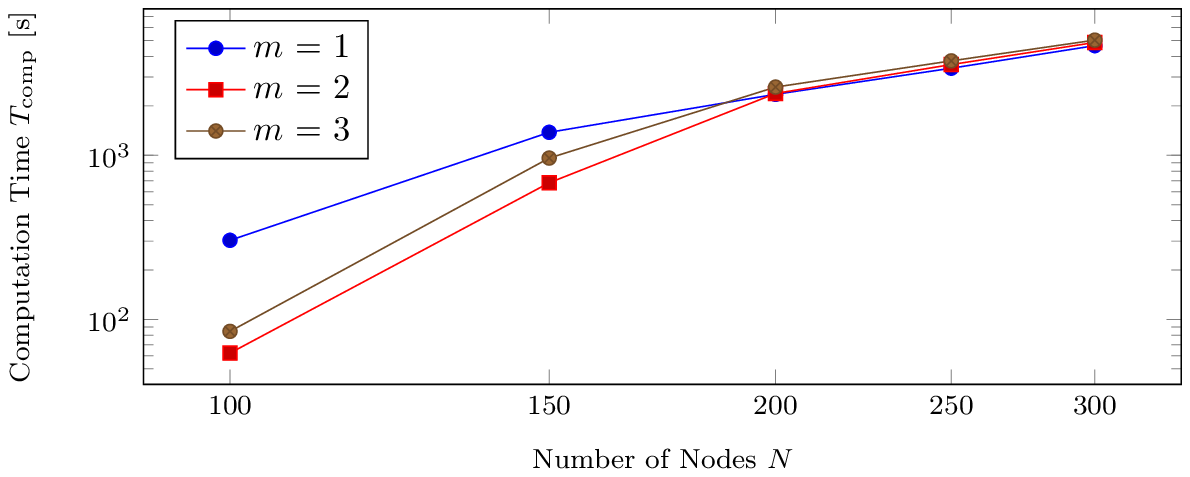}
	{\footnotesize\caption{The computation time $T_\textrm{comp}$ of Algorithm \ref{ALG_CLASSO} versus the number of nodes $N$ on a log-log scale, averaged over 50 networks which were generated by the Barab\'asi-Albert random graph model with sparsity parameter $m = 1, 2, 3$. \label{fig_copm_time_BA}}}
\end{figure}

\section{Conclusions}
\label{sec:conclusion}
This work considers the prediction of the viral dynamics between groups of individuals which is described by the discrete-time NIMFA model. Our contribution is twofold. 

First, we give an alternative and equivalent description of the NIMFA equations. We show that, if the initial viral states are close to zero, then the viral states of every group are upper bounded by the respective steady-state. Hence, the absence of overshooting the steady-state, which is a reasonable assumption for many real spreading phenomena, is captured by the NIMFA model.

Second, we introduced a network reconstruction and spreading parameter estimation algorithm by formulating the estimation problem in a basis pursuit sense. By numerical simulations, we show that the algorithm has a polynomial runtime. Sparse networks can be reconstructed accurately from only one epidemic outbreak, and dense network can be inferred by observing multiple outbreaks, as demonstrated by numerical evaluations. A powerful application of the network reconstruction method is the prediction of the viral state evolution of every node, given only a few observations on an unknown sparse network.

For very large-scale networks, developing faster, and in particular distributed, versions the proposed network reconstruction algorithm stands on the agenda of future research. Furthermore, our ambition is to develop (possibly approximate) methods for predicting the viral state evolution for an arbitrarily dense contact network.

\bibliographystyle{ieeetr}     

\appendix

\section{Proof of Proposition \ref{proposition:delta_v_equation}}
\label{appendix:difference_eq}
Since $\Delta v_i[k+1]= v_i[k+1] -v_{\infty, i}$, the evolution of the difference $\Delta v_i[k]$ over time $k$ can be stated with the NIMFA equations (\ref{NIMFA_disc_}) as
\begin{align}\label{delta_vi_ss}
\Delta v_i[k+1] & = (1 - \delta_T) v_i[k] + \beta_T \sum^N_{j=1} a_{ij}v_j[k] -\beta_T v_i[k] \sum^N_{j=1} a_{ij}v_j[k]  - v_{\infty, i}
\end{align}
We would like to express the difference $\Delta v_i[k+1]$ at the next time $k+1$ only in dependency of the difference $\Delta v[k]$ at the current time $k$ and the constant steady state $v_\infty$. The steady state $v_\infty$ is an equilibrium point of (\ref{NIMFA_disc_stacked}) and satisfies
\begin{align} \label{steadyState_f}
v_{\infty ,i} = (1 - \delta_T)v_{\infty, i} + \beta_T (1 - v_{\infty,i}) \sum^N_{j=1} a_{ij} v_{\infty, j},
\end{align}
for all nodes $i$. We insert (\ref{steadyState_f}) in (\ref{delta_vi_ss}) and obtain
\begin{align} \label{oijoijoijo}
\Delta v_i[k+1] & = (1 - \delta_T) v_i[k] + \beta_T \sum^N_{j=1} a_{ij}v_j[k]  - \beta_T v_i[k] \sum^N_{j=1} a_{ij}v_j[k]   \\
&\quad -(1 - \delta_T)v_{\infty, i} - \beta_T  \sum^N_{j=1} a_{ij} v_{\infty, j} +\beta_T v_{\infty, i}\sum^N_{j=1} a_{ij} v_{\infty, j} \nonumber
\end{align}
Since $\Delta v_i[k] = v_i[k]- v_{\infty,i}$, we can express (\ref{oijoijoijo}) more compactly as
\begin{align} \label{eqljknlnklsdgf}
\Delta v_i[k+1] & = (1 - \delta_T)\Delta v_i[k]+ \beta_T  \sum^N_{j=1} a_{i j} \Delta v_j[k]  -\beta_T   \sum^N_{j=1}a_{i j} \left( v_i[k]v_j[k] - v_{\infty, i} v_{\infty, j} \right)
\end{align}
The first two addends in (\ref{eqljknlnklsdgf}) are already in the desired form: they depend on the difference $\Delta v[k]$ but not on the viral state $v[k]$ at time $k$. To replace the viral state $v[k]$ in the last term of (\ref{eqljknlnklsdgf}) by an expression of the difference $\Delta v[k]$, we observe that
\begin{align} \label{sljknkljnd}
 v_i[k]v_j[k] - v_{\infty, i} v_{\infty, j}=  \Delta v_i[k] \Delta v_j[k] +  \Delta v_i[k] v_{\infty, j} + v_{\infty, i}  \Delta v_j[k], 
\end{align}
since $v_i[k] = \Delta v_i[k]+ v_{\infty,i}$. Inserting (\ref{sljknkljnd}) in (\ref{eqljknlnklsdgf}) yields 
\begin{align} \label{ljnljjjjsssafff}
\Delta v_i[k+1] & = \left( 1 - \delta_T  -\beta_T  \sum^N_{j=1}a_{i j}  v_{\infty, j} \right)\Delta v_i[k] + \beta_T (1- v_{\infty, i} )  \sum^N_{j=1} a_{i j} \Delta v_j[k]  -\beta_T   \Delta v_i[k] \sum^N_{j=1}a_{i j}  \Delta v_j[k]   
\end{align}
The expression (\ref{ljnljjjjsssafff}) can be further simplified. The steady-state equation (\ref{steadyState_f}) is equivalent to  
\begin{align} \label{stads}
\sum^N_{j=1} a_{ij} v_{\infty, j} = \frac{\delta_T}{\beta_T}\left( \frac{1}{1 - v_{\infty, i}} -1\right)
\end{align}
From (\ref{stads}) follows that (\ref{ljnljjjjsssafff}) is equivalent to 
\begin{align} \label{delta_vi_single}
\Delta v_i[k+1] & = \left( 1  + \frac{\delta_T}{ v_{\infty, i} - 1}\right) \Delta v_i[k]+ \beta_T (1- v_{\infty, i} )  \sum^N_{j=1} a_{i j} \Delta v_j[k]  -\beta_T   \Delta v_i[k] \sum^N_{j=1}a_{i j}  \Delta v_j[k]   
\end{align}
Stacking equation (\ref{delta_vi_single}) for all nodes $i = 1, ..., N$ completes the proof.

\section{Proof of Corollary \ref{corollary:below_steady_state}}
\label{appendix:below_steady_state}
We rewrite equation (\ref{ljnljjjjsssafff}) to obtain
\begin{align} \label{lkjnkja}
\Delta v_i[k+1] & = g_i[k] + h_i[k] \Delta v_i[k] 
\end{align}
where the terms $g_i[k]$ and $h_i[k]$ are given by 
\begin{align*}
g_i[k] = \beta_T (1- v_{\infty, i} )  \sum^N_{j=1} a_{i j} \Delta v_j[k]    
\end{align*}
and
\begin{align*}
h_i[k] = \left( 1 - \delta_T  -\beta_T  \sum^N_{j=1}a_{i j}  \left( v_{\infty, j} +\Delta v_j[k]\right) \right)
\end{align*}
for every node $i$. Since $a_{i j} \ge 0$, $\beta_T > 0$ and $(1- v_{\infty, i} ) \ge 0$, it holds that
\begin{align*}
\Delta v_j[k] \ge 0 ~ \forall j \Rightarrow g_i[k]\ge 0
\end{align*}
Furthermore, it holds $ v_{\infty, j} +\Delta v_j[k] =  v_j[k]$ and $ v_j[k] \in [0,1]$, which yields
\begin{align} \label{klkjnlljnklnnlkn}
h_i[k] = 1 - \delta_T  -\beta_T  \sum^N_{j=1}a_{i j}  v_j[k] \ge 1 - \delta_T -n\beta_T  d_i 
\end{align}
Thus, if $\delta_T  +\beta_T  d_i \le 1$, then (\ref{klkjnlljnklnnlkn}) implies that $h_i[k] \ge 0$. From $g_i[k] \ge 0$ and $h_i[k] \ge 0$ if $\Delta v_i[k] \ge 0$ for all nodes $i$ and (\ref{lkjnkja}) follows that: $\Delta v_i[k+1] \ge 0$ for all nodes $i$ if $\Delta v_i[k] \ge 0$ for all nodes $i$. Hence, we obtain by induction that $\Delta v_i[0] \ge 0$ for all nodes $i$ implies $\Delta v_i[k] \ge 0$ for all nodes $i$ at every time $k$, which proves Corollary \ref{corollary:below_steady_state}. (Analogously, we can prove that $\Delta v_i[0] \le 0$ for all nodes $i$ implies $\Delta v_i[k] \le 0$ for all nodes $i$ at every time $k$.)

\section{Proof of Lemma \ref{lemma:reconstruction_as_linear_system}}
\label{appendix:reconstruction_as_linear_system}
We reformulate the systems equations (\ref{NIMFA_disc_}) for node $i$ as
\begin{align*}
v_i [k + 1] - v_i[k]  & =  - \delta_T v_i[k] +  (1 - v_i[k]) v^T[k] \beta_T A_i
\end{align*}
Dividing by $(1 - v_i[k])$ yields
\begin{align} \label{eq:lemma_proof_a}
\frac{v_i [k + 1] - v_i[k] }{1 - v_i[k]}  & =  \delta_T \frac{v_i[k]}{ v_i[k] - 1} +  v^T[k] \beta_T A_i
\end{align}
With (\ref{eq:def_bik}) and (\ref{eq:def_cik}), we obtain from (\ref{eq:lemma_proof_a}) that
\begin{align}
b_i[k]  & =   \delta_T c_i[k] +  v^T[k] \beta_T A_i \label{ci_bi_NIMFA}
\end{align}
For a single node $i$, we stack equation (\ref{ci_bi_NIMFA}) for time instants $k= 0, ..., n-1$ and obtain
\begin{align*}
\begin{pmatrix}
b_i[0] \\
\vdots\\
b_i[n-1]
\end{pmatrix} & =    \begin{pmatrix}
c_i[0] \\
\vdots\\
c_i[n-1]
\end{pmatrix} \delta_T +  \begin{pmatrix}
v^T[0] \\
\vdots\\
v^T[n-1]
\end{pmatrix} \beta_T A_i 
\end{align*}
By combining the $i$-th column of the weighted adjacency matrix $\beta_T A_i$ and the curing rate $\delta_T$ into one vector, we obtain
\begin{align}
\begin{pmatrix}
b_i[0] \\
\vdots\\
b_i[n-1]
\end{pmatrix} & =  \begin{pmatrix}
v^T[0] & c_i[0] \\
\vdots & \vdots \\
v^T[n-1] & c_i[n-1]
\end{pmatrix} \begin{pmatrix}
 \beta_T A_i\\
\delta_T
\end{pmatrix} \label{eq:bvA}
\end{align}
With (\ref{matriecsBCV}), we rewrite (\ref{eq:bvA}) as
\begin{align}
B_i & =  \begin{pmatrix}
V & C_i
\end{pmatrix} \begin{pmatrix}
\beta_T A_i\\
\delta_T
\end{pmatrix} , \label{sdfddd}
\end{align}
where the vectors $B_i \in \mathbb{R}^{n}$ and $C_i\in \mathbb{R}^{n}$ denote the $i$-th column of the matrix $B$ and the matrix $C$, respectively. Stacking the equation (\ref{sdfddd}) for row $i = 1, ..., N$ yields
\begin{align*}
\begin{pmatrix}
B_1 \\
\vdots \\
B_N
\end{pmatrix}  & =  \begin{pmatrix}
V & 0 & ...& 0& C_1\\
0& V  & ...&0& C_2\\
\vdots &   &\ddots&& \vdots\\
 0&    ... &0&V & C_N
\end{pmatrix} \begin{pmatrix}
\beta_T A_1 \\
\vdots \\
\beta_T A_N\\
\delta_T
\end{pmatrix} 
\end{align*}
We define the $Nn \times 1$ vectors $B_\textrm{vec} = (B^T_1, ..., B^T_N)^T$ and $C_\textrm{vec} = (C^T_1, ..., C^T_N)^T$. Finally, we obtain the set of linear equations 
\begin{align*}
B_\textrm{vec}  & =  \begin{pmatrix}
I_N \otimes V & C_\textrm{vec} 
\end{pmatrix} \begin{pmatrix}
\beta_T A_\textrm{vec} \\
\delta_T
\end{pmatrix}, 
\end{align*}
where $\otimes$ denotes the Kronecker product of two matrices, which proves Lemma \ref{lemma:reconstruction_as_linear_system}.

\section{Proof of Lemma \ref{lemma:reduced_size_lin_sys}}
\label{appendix:reduced_size_lin_sys}
We denote the singular value decomposition (SVD) of the viral state matrix $V$ by
\begin{align} \label{V_usvd}
V = U_r S_r Q^T_r + U_\epsilon S_\epsilon Q^T_\epsilon,
\end{align}
where $S_r = \textrm{diag}(\sigma_1, ..., \sigma_r)$ is the $r \times r$ diagonal matrix of the largest $r$ singular values and $S_\epsilon$ is the $(n-r) \times (n-r)$diagonal matrix of the remaining $(n-r)$ singular values. The $n \times r$ matrix $U_r$ and the $n \times (n-r)$ matrix $U_\epsilon$ contain the respective left-singular vectors, and the $N \times r$ matrix $Q_r$ and the $N\times (n-r)$ matrix $Q_\epsilon$ contain the respective right-singular vectors.

We insert the SVD (\ref{V_usvd}) of the matrix $V$ in the linear system (\ref{eq:lin_sys}) and obtain for every $l= 1, ..., N$ that
\begin{align} \label{eq:linsys_svd}
B_l = U_r S_r Q^T_r \xi_l + U_\epsilon S_\epsilon Q^T_\epsilon \xi_l + C_l \delta_T,
\end{align}
where the $N \times 1$ vector $\xi_l$ equals the $l$-th part of the $N^2 \times 1$ vector $\Pi_N w$, when $\Pi_N w$ is divided into $N$ equally sized parts, i.e., $\Pi_N w = (\xi^T_1, ..., \xi^T_N)^T$. Since $U_\epsilon S_\epsilon Q^T_\epsilon \xi_l \approx 0$, we can approximate the linear system (\ref{eq:linsys_svd}) by 
\begin{align} \label{eq:linsys_svd_b}
B_l = U_r S_r Q^T_r \xi_l + C_l \delta_T
\end{align}
for every $l= 1, ..., N$. Finally, we replace (\ref{eq:linsys_svd_b}) by 
\begin{align} \label{eq:linsys_svd_if}
\tilde{B}_l =  Q^T_r \xi_l + \tilde{C}_l \delta_T,
\end{align}
where the two $r \times 1$ vectors $\tilde{B}_l$, $\tilde{C}_l$ are obtained as follows. First, the two $n \times 1$ vectors $B_l$ and $C_l$ are projected onto the image of the singular vector matrix $U_r$, which gives two $n \times 1$ vectors $B_{l, U}$ and $C_{l, U}$. Then, the two $r \times 1$ vectors $\tilde{B}_l$ and $\tilde{C}_l$ follow as the solution to 
\begin{align*}
B_{l, U} = U_r S_r \tilde{B}_l \quad \text{and} \quad C_{l, U} = U_r S_r \tilde{C}_l 
\end{align*}
The two $r \times 1$ vectors $\tilde{B}_l$, $\tilde{C}_l$ are explicitly given by 
\begin{align*}
\tilde{B}_l = \left(U_r S_r \right)^\dagger B_l \quad \text{and} \quad \tilde{C}_l \left(U_r S_r \right)^\dagger C_l,
\end{align*}
where $\left(U_r S_r \right)^\dagger$ denotes the Moore-Penrose pseudoinverse of the matrix $U_r S_r$. Stacking (\ref{eq:linsys_svd_if}) for $l=1, ..., N$ yields the linear system
\begin{align*}
\begin{pmatrix}
\tilde{B}_1 \\
\tilde{B}_2 \\
\vdots \\
\tilde{B}_N
\end{pmatrix}  & =  \begin{pmatrix}
Q^T_r & 0 & ...& 0& \tilde{C}_1\\
0& Q^T_r  & ...&0& \tilde{C}_2\\
\vdots &   &\ddots&& \vdots\\
 0&    ... &0&Q^T_r & \tilde{C}_N
\end{pmatrix} 
\begin{pmatrix} 
\Pi_N & 0 \\
0 & 1
\end{pmatrix}
\begin{pmatrix} 
w \\
\delta_T
\end{pmatrix},
\end{align*}
which proves Lemma \ref{lemma:reduced_size_lin_sys}.
\end{document}